\documentclass[11pt]{article}             
\usepackage{osid}   
\usepackage{cite}
\usepackage{extract}
\usepackage{fancyhdr}
\usepackage[utf8]{inputenc}
\usepackage{amsmath, amsfonts, amssymb, amstext, amsfonts}
\usepackage{graphicx}
\usepackage{multirow}
\usepackage{wrapfig}
\usepackage{epsfig}
\usepackage{framed}
\usepackage{textcomp}
\usepackage{dsfont}
\usepackage{tikz}
\usepackage[ruled,vlined]{algorithm2e}
\usepackage{booktabs}
\usepackage[left=3.82cm,right=3.82cm]{geometry} 

\usepackage{graphicx}
\usepackage{dcolumn}
\usepackage{bm}
\usepackage{comment}
 \usepackage{amsthm}

 \usepackage{mathtools}
 \usepackage{color}
 \usepackage{enumitem}
 \usepackage{verbatim}
 \usepackage[
	colorlinks=true,
	urlcolor=blue,
	linkcolor=green
]{hyperref}

\newcommand{\idmat}{\mathds{1}}
\newcommand{\unitary}{\mathcal{U}}

\newcommand{\idop}{\operatorname{id}}
\newcommand{\blt}{\mathcal{L}}

\newcommand{\trcl}{\mathcal{S}_1}

\newcommand{\norm}[1]{\left\lVert #1 \right\rVert}

\newcommand{\tr}[1]{\mathrm{tr}\left[#1\right]} 
\newcommand{\ptr}[2]{\mathrm{tr}_{#1}\left[#2\right]}

\newcommand{\CP}{\mathrm{CP}_\sigma}
\MakeRobust{\ref}

\makeatletter
\newcommand{\labeltext}[2]{%
  #1
  \@bsphack
  \csname phantomsection\endcsname 
  \def\@currentlabel{#1}{\label{#2}}%
  \@esphack
}
\makeatother

\DeclarePairedDelimiter\bra{\langle}{\rvert}
\DeclarePairedDelimiter\ket{\lvert}{\rangle}
\DeclarePairedDelimiterX\braket[2]{\langle}{\rangle}{#1 \delimsize\vert #2}

\title{On the generators of quantum dynamical semigroups with invariant subalgebras}
\author{Markus Hasenöhrl\thanks{~Supported by the Bavarian excellence network {\sc enb} via the \mbox{International} PhD Programme of Excellence {\em Exploring Quantum Matter.}}\\{\footnotesize\it Department of Mathematics, Technical University of Munich, 85748 Garching, Germany\\
Munich Center for Quantum Science and Technology (MCQST), 80799 Munich, Germany\\
m.hasenoehrl@tum.de}\\[2ex]
        Matthias C.~Caro\thanks{~Supported by the TopMath Graduate Center of the TUM Graduate School at the
Technische Universität München, Germany, the TopMath Program at the Elite Network of Bavaria, and the German Academic Scholarship Foundation (Studienstiftung des deutschen Volkes).}
                     \\{\footnotesize\it Department of Mathematics, Technical University of Munich, 85748 Garching, Germany\\
                     Munich Center for Quantum Science and Technology (MCQST), 80799 Munich, Germany\\
                     caro@ma.tum.de} }

\begin{document}

\maketitle
\begin{abstract}
     The problem of characterizing GKLS-generators and CP-maps with an invariant appeared in different guises in the literature.
     We prove two unifying results which hold even for  weakly closed *-algebras: First, we show how to construct a normal form for $\mathcal{A}$-invariant GKLS-generators, if a normal form for $\mathcal{A}$-invariant CP-maps is known --- rendering the two problems essentially equivalent. 
     Second, we provide a normal form for $\mathcal{A}$-invariant CP-maps if $\mathcal{A}$ is atomic (which includes the finite-dimensional case). 
     As an application we reproduce several results from the literature as direct consequences of our characterizations and thereby point out connections between different fields.
\end{abstract}

\section{Introduction}

Quantum dynamical semigroups play in important role in many areas of physics. A (norm-continuous) quantum dynamical semigroup is a collection of normal completely positive maps $(T_t)_{t \geq 0}$ on $\blt(\mathcal{H})$ such that $T_0 = \idop$, $T_{s + t} = T_s \circ T_t$ for all $s, t \geq 0$ and the map $t \mapsto T_t$ is norm-continuous. By the general theory of continuous one-parameter semigroups \cite[Theorem 3.7]{engel2006one}, there exists a bounded operator $L$, called generator, such that $T_t = e^{tL}$ for all $t \geq 0$. The fundamental result due to Gorini, Kossakowski, Sudarshan \cite{Gorini.1976} and Lindblad \cite{Lindblad.1976} is that $L$ generates a norm-continuous quantum dynamical semigroup if and only if $L$ is of the form  
\begin{align}
    L(X) = V^\dagger (X \otimes \idmat_\mathtt{E}) V - K^\dagger X - XK, \quad X \in \blt(\mathcal{H}),
\end{align}
for some $V \in \blt(\mathcal{H}; \mathcal{H} \otimes \mathcal{H}_\mathtt{E})$ and $K \in \blt(\mathcal{H})$.

In the past, special cases of the following question on restricted GKLS-generators arose in the literature: Suppose $\mathcal{A} \subseteq \blt(\mathcal{H})$ is a von Neumann algebra, or more generally a weakly closed *-algebra\footnote{A weakly closed *-algebra can be characterized as a von Neumann algebra with support only on a subspace. It therefore consists of a von Neumann algebra and a null part. As the step from von Neumann algebras to weakly closed *-algebras is not always obvious in our results, we include the null part throughout for completeness.}, such that $T_t(\mathcal{A}) \subseteq \mathcal{A}$ for all $t \geq 0$ or, equivalently, such that $L(\mathcal{A}) \subseteq \mathcal{A}$. How does this condition constrain the operators $V$ and $K$? 
In this work, we provide an answer to this question, if $\mathcal{A}$ is atomic --- thus covering many interesting cases, in particular the finite-dimensional case. It should be noted here that $T_t(\mathcal{A}) \subseteq \mathcal{A}$ for all $t \geq 0$  is equivalent to $L(\mathcal{A}) \subseteq \mathcal{A}$.

Among the results for which the answer to the question posed above is useful are: The Koashi-Imoto theorem~\cite{PhysRevA.66.022318}, an important result in the theory of quantum communication, giving the general form of a quantum channel leaving a certain set of density matrices invariant; the form of the GKLS-generator imposed by the invariance of the decoherence-free subalgebra \cite{doi:10.1142/S0129055X16500033, doi:10.1063/1.5030954, doi:10.1142/S0219025710004176, sasso2021general}; general questions about decoherence, where the form of the GKLS-generator imposed by the invariance of a maximally abelian subalgebra is important \cite{doi:10.1142/S0129055X03001631, REBOLLEDO2005349, rajarama2008maximal, doi:10.1007/s11080-005-0485-3, FrancoFagnola2007}; the study of Markovian subsystems \cite{ticozzi2008quantummarkovian} and the study of the aging process of quantum devices via dynamical semigroups of superchannels \cite{hasenoehrl2021quantum}.

This paper is structured as follows: In Section \ref{PreliminarySection}, we introduce the notation and remind the reader of several facts related to completely positive (CP) maps, GKLS-generators, and weakly closed *-algebras. In Section \ref{GKLS-InvariantSection}, we show how to reduce the general problem of classifying GKLS-generators with an invariant approximately finite-dimensional algebra to the one of classifying normal CP-maps with the same invariant algebra. In Section \ref{CP-invarinatSection} we classify normal completely positive maps with an invariant \emph{atomic} algebra. Section \ref{AtomicInvarinatGKSL} combines the results from Sections \ref{GKLS-InvariantSection} and \ref{CP-invarinatSection} to obtain a classification of GKLS-generators with invariant atomic algebras. In Section \ref{applicationSection}, we use our results to reproduce several results from the literature discussed above. Finally, in Section \ref{ConclusionSection}, we conclude our work and outline possible further lines of research.  

\section{Preliminaries and Notation} \label{PreliminarySection}

\paragraph{Functional analysis:}

Throughout, $\mathcal{H}$ (with some subscript) denotes a \textit{separable} complex Hilbert space. For Banach spaces $\mathcal{X}$ and $\mathcal{Y}$, we denote by $\blt(\mathcal{X}; \mathcal{Y})$ the set of bounded linear operators from $\mathcal{X}$ to $\mathcal{Y}$, which becomes a Banach space when equipped with the operator norm. 
We abbreviate $\blt(\mathcal{X}; \mathcal{X})$ by $\blt(\mathcal{X})$. The \emph{strong operator topology (SOT)} on $\blt(\mathcal{H}_\mathtt{A}; \mathcal{H}_\mathtt{B})$ is the smallest topology such that for all $\ket{\psi_\mathtt{A}} \in \mathcal{H}_\mathtt{A}$, the map $X \mapsto X\ket{\psi_\mathtt{A}}$ is continuous. The \emph{weak operator topology (WOT)} on $\blt(\mathcal{H}_\mathtt{A}; \mathcal{H}_\mathtt{B})$ is the smallest topology such for all $\ket{\psi_\mathtt{A}} \in \mathcal{H}_\mathtt{A}$ and $\ket{\psi_\mathtt{B}} \in \mathcal{H}_\mathtt{B}$, the map $X \mapsto \braket{\psi_\mathtt{B}}{X\psi_\mathtt{A}}$ is continuous. The \emph{ultraweak (or weak-*) topology} on $\blt(\mathcal{H}_A; \mathcal{H}_B)$ is the smallest topology such that for all $\rho \in \trcl(\mathcal{H}_B; \mathcal{H}_A)$, the map $X \mapsto \tr{X\rho}$ is continuous. Here, $\mathrm{tr}$ denotes the trace and $\trcl(\mathcal{H}_B; \mathcal{H}_A)$ the set of \emph{trace-class operators}, that is those $\rho \in \blt(\mathcal{H}_B; \mathcal{H}_A)$ for which $\tr{\sqrt{\rho^\dagger \rho}}  < \infty$. 
A subset $T \subseteq \mathcal{H}$ is \emph{total} in $\mathcal{H}$ if its linear span is dense in $\mathcal{H}$. An operator $V \in \blt(\mathcal{H}_\mathtt{A}; \mathcal{H}_\mathtt{B})$ is called an \emph{isometry} if $\norm{V\ket{\psi}} = \norm{\ket{\psi}}$ for all $\ket{\psi} \in \mathcal{H}_\mathtt{A}$. A surjective isometry is called \emph{unitary}.  

\paragraph{CP-maps and GKLS-generators:}

A linear map $\Phi : \blt(\mathcal{H}_\mathtt{A}) \rightarrow \blt(\mathcal{H}_\mathtt{B})$ is a \textit{normal CP-map} if there exists a Hilbert space $\mathcal{H}_\mathtt{E}$ and an operator $V \in \blt(\mathcal{H}_\mathtt{B}; \mathcal{H}_\mathtt{A} \otimes \mathcal{H}_\mathtt{E})$ such that $\Phi(X) = V^\dagger (X \otimes \idmat_\mathtt{E}) V$. We denote the set of such normal CP-maps by $\CP(\mathcal{H}_\mathtt{A}; \mathcal{H}_\mathtt{B})$ and abbreviate $\CP(\mathcal{H}; \mathcal{H})$ by $\CP(\mathcal{H})$. The pair $(V, \mathcal{H}_\mathtt{E})$ is called a \emph{Stinespring representation} of $\Phi$. An equivalent characterization of normal CP-map is that they admit a \emph{Kraus representation}. That is, there exist operators $\{v_i\}_i  \subset \blt(\mathcal{H})$ such that $\Phi(X) = \sum_i v_i^\dagger X v_i$ for all $X \in \blt(\mathcal{H})$, where the series is SOT-convergent. The choice of $(V, \mathcal{H}_\mathtt{E})$ representing $\Phi$ is not unique. However, the following well-known theorem (see e.g. \cite[Theorem 29.6]{k1992introduction}) quantifies the freedom.  

\begin{theorem}{Theorem} \label{thm:UniquenessStinespringThm}
Let $\tilde{V} \in \blt(\mathcal{H}_\mathtt{B}; \mathcal{H}_\mathtt{A} \otimes \mathcal{H}_{\tilde{\mathtt{E}}})$ and define $\Phi(X) = \tilde{V}^\dagger (X \otimes \idmat_{\tilde{\mathtt{E}}}) \tilde{V}$ for all $X \in \blt(\mathcal{H}_\mathtt{A})$. Then there exist $\mathcal{H}_\mathtt{E}$ and $V \in \blt(\mathcal{H}_\mathtt{B}; \mathcal{H}_\mathtt{A} \otimes \mathcal{H}_\mathtt{E})$ such that \labeltext{a)}{eq:StinespringUniquenessStatement} $\Phi(X) = V^\dagger (X \otimes \idmat_\mathtt{E}) V$ for all $X \in \blt(\mathcal{H}_\mathtt{A})$ and such that \labeltext{b)}{eq:StinespringMinimality} $\{ (X \otimes \idmat_\mathtt{E})V \ket{\psi} \, | \, X \in \blt(\mathcal{H}_\mathtt{A}), \ket{\psi} \in \mathcal{H}_\mathtt{B} \}$ is total in $\mathcal{H}_\mathtt{A} \otimes \mathcal{H}_\mathtt{E}$.\\ 
If $(\mathcal{H}_\mathtt{E}, V)$ is any pair such that \ref{eq:StinespringUniquenessStatement} and \ref{eq:StinespringMinimality} are satisfied, and if $(\mathcal{H}_{\tilde{\mathtt{E}}}, \tilde{V})$ is another pair such that \ref{eq:StinespringUniquenessStatement} is satisfied, then there exists an isometry $W \in \blt(\mathcal{H}_\mathtt{E}; \mathcal{H}_{\tilde{\mathtt{E}}})$ such that $\tilde{V} = (\idmat_\mathtt{A} \otimes W)V$. If \ref{eq:StinespringMinimality} is also satisfied for $(\mathcal{H}_{\tilde{\mathtt{E}}}, \tilde{V})$, then $W$ is unitary. 
\end{theorem}

\noindent If $V$ satisfies condition \ref{eq:StinespringMinimality} above, then it is called \textit{minimal}. 

A linear map $L : \blt(\mathcal{H}) \rightarrow \blt(\mathcal{H})$ is called \textit{GKLS-generator} (or generator in GKLS-form) if there exists $\Phi \in \CP(\mathcal{H})$ and $K \in \blt(\mathcal{H})$, such that $L(X) = \Phi(X) - K^\dagger X - XK$ for all $X \in \blt(\mathcal{X})$. As for normal CP-maps, the representation is not unique. The following characterization of the freedom can be extracted from \cite[Chapter 30]{k1992introduction}, in particular from the proof of Proposition 30.14. We give a complete proof in Appendix \ref{ProofUniquenessAppendix}. 
\begin{theorem}{Theorem} \label{thm:GeneralUniquenessTheorem}
Let $\tilde{V} \in \blt(\mathcal{H}; \mathcal{H} \otimes \mathcal{H}_{\tilde{\mathtt{E}}})$ and $\tilde{K} \in \blt(\mathcal{H})$ and define $L(X) = \tilde{V}^\dagger (X \otimes \idmat_{\tilde{\mathtt{E}}}) \tilde{V} - \tilde{K}^\dagger X - X\tilde{K}$ for all $X \in \blt(\mathcal{H})$. Then there exist $\mathcal{H}_\mathtt{E}$, $V \in \blt(\mathcal{H}; \mathcal{H} \otimes \mathcal{H}_\mathtt{E})$ and $K \in \blt(\mathcal{H})$ such that \labeltext{a)}{eq:GKLSUniquenessStatement} $L(X) = V^\dagger (X \otimes \idmat_\mathtt{E}) V - K^\dagger X - XK$ for all $X \in \blt(\mathcal{H})$ and such that \labeltext{b)}{eq:GKLSMinimality} 
    $\{((X \otimes \idmat_\mathtt{E})V - VX)\ket{\psi} \, | \, X \in \blt(\mathcal{H}), \ket{\psi} \in \mathcal{H} \}$  is total in $\mathcal{H} \otimes \mathcal{H}_\mathtt{E}$.\\
If $(\mathcal{H}_\mathtt{E}, V, K)$ is any triplet such that \ref{eq:GKLSUniquenessStatement} and \ref{eq:GKLSMinimality} are satisfied, and if $(\mathcal{H}_{\tilde{\mathtt{E}}}, \tilde{V}, \tilde{K})$ is another triplet such that \ref{eq:GKLSUniquenessStatement} is satisfied, then there exists an isometry $W \in \blt(\mathcal{H}_\mathtt{E}; \mathcal{H}_{\tilde{\mathtt{E}}})$, a vector $\ket{\tilde{\psi}} \in \mathcal{H}_{\tilde{\mathtt{E}}}$, and a number $\mu \in \mathbb{R}$ such that
\begin{align} \label{eq:FreedomSimpleLindbladGenerator}
    \tilde{V} = (\idmat \otimes W)V + \idmat \otimes \ket{\tilde{\psi}}, \quad  \tilde{K} = K + (\idmat \otimes \bra{ \tilde{\psi}}W)V + \frac{1}{2} \Vert\tilde{\psi}\Vert^2 + i\mu.
\end{align}
If \ref{eq:GKLSMinimality} is also satisfied for $(\mathcal{H}_{\tilde{\mathtt{E}}}, \tilde{V}, \tilde{K})$, then $W$ is unitary. 
\end{theorem}

\paragraph{Weakly closed *-algebras:} \label{sec:NotationW*}

We introduce several conventions that will be useful in simplifying the notation throughout.  
A \emph{weakly closed *-algebra} $\mathcal{A} \subseteq \blt(\mathcal{H})$ is a subalgebra of $\blt(\mathcal{H})$ that is closed w.r.t.~the WOT\footnote{Equivalently, one can use the SOT or the ultraweak topology.} and w.r.t.~taking adjoints. A weakly closed *-algebra does not necessarily contain the identity  --- if it does then it is called a \emph{von Neumann-algebra} (abbr.~\emph{vN-algebra}). 
Every weakly closed *-algebra $\mathcal{A}$ is (unitarily equivalent to) the direct sum of a zero-dimensional algebra $0_0$ and a vN-Algebra \cite[Proposition 5.1.8]{kadison1997fundamentals}. That is, $\mathcal{A} = U_\mathcal{A} \left( 0_0 \oplus \mathcal{A}_{\overline{0}} \right) U^\dagger_\mathcal{A}$, where $U_\mathcal{A} : \mathcal{H}_\oplus \rightarrow \mathcal{H}$ is a unitary on $\mathcal{H}_\oplus = \mathcal{H}_0 \oplus \mathcal{H}_{\overline{0}}$, $0_0 = \{0 \} \subseteq \blt(\mathcal{H}_0)$ and $\mathcal{A}_{\overline{0}}$ is a vN-algebra in $\blt(\mathcal{H}_{\overline{0}})$. If $P_0^\oplus \in \blt(\mathcal{H}_\oplus; \mathcal{H}_0)$ and $P_{\overline{0}}^\oplus \in \blt(\mathcal{H}_\oplus; \mathcal{H}_{\overline{0}})$ are the orthogonal projections onto $\mathcal{H}_0$ and $\mathcal{H}_{\overline{0}}$, then we define $P_0 \in \blt(\mathcal{H}; \mathcal{H}_0)$ and $P_{\overline{0}} \in \blt(\mathcal{H}; \mathcal{H}_{\overline{0}})$ by $P_0 = P_0^\oplus U_\mathcal{A}^\dagger$ and $P_{\overline{0}} = P_{\overline{0}}^\oplus U_\mathcal{A}^\dagger$.

Two special types of weakly closed *-algebras are of particular importance to us: the approximately finite-dimensional ones and the atomic ones.
A weakly closed *-algebra $\mathcal{A} \subseteq \blt(\mathcal{H})$ is called \emph{approximately finite-dimensional (AFD)} if there exists an increasing sequence $\mathcal{A}_1 \subseteq \mathcal{A}_2 \subseteq \mathcal{A}_3 \subseteq \cdots \subseteq \mathcal{A}$ of finite-dimensional (hence weakly closed) sub-*-algebras of $\mathcal{A}$ such that $\cup_{n \in \mathbb{N}} \mathcal{A}_n$ is WOT-dense in $\mathcal{A}$. 
\emph{Atomic} weakly closed *-algebras are usually defined by the requirement that every non-zero projection in $\mathcal{A}$ majorizes a non-zero minimal projection \cite[Definition 5.9]{takesaki1979theory} --- a property always fulfilled in finite dimensions. For our purposes, it is more convenient to think of them as those weakly closed *-algebras that are the direct sum of type-I factors. A proof of this equivalence can be found in the appendix of \cite{doi:10.1142/S0129055X16500033}.

\begin{definition}{Definition} \label{Defn:NormalFormAtomicAlgebra}
A weakly closed *-algebra $\mathcal{A} \subseteq \blt(\mathcal{H})$ is called \textit{atomic} if
    \begin{align}
        \mathcal{A} = U_\mathcal{A} \left( 0_0 \oplus \bigoplus_{i \in I}  (\blt(\mathcal{H}_{\mathtt{A}_i}) \otimes \idmat_{\mathtt{B}_i})\right) U_{\mathcal{A}}^\dagger,
    \end{align}
    for a Hilbert space $\mathcal{H}_{0}$, sequences of Hilbert spaces $\{\mathcal{H}_{A_i}\}_{i \in I}$ and $\{\mathcal{H}_{B_i}\}_{i \in I}$ indexed by a countable index set $I$, and a unitary $U_\mathcal{A} : \mathcal{H}_\oplus \rightarrow \mathcal{H}$, where $\mathcal{H}_\oplus = \mathcal{H}_{0} \oplus \bigoplus_{i \in I} (\mathcal{H}_{A_i} \otimes \mathcal{H}_{B_i})$. 
    
    We further define for all $i \in I$ the Hilbert space $\mathcal{H}_i = \mathcal{H}_{\mathtt{A}_i} \otimes \mathcal{H}_{\mathtt{B}_i}$. For all $k \in I \cup \{0\}$, let $P_k^\oplus \in \blt(\mathcal{H}_\oplus; \mathcal{H}_k)$ be the orthogonal projection onto $\mathcal{H}_k$ and let us define $P_k \in \blt(\mathcal{H}; \mathcal{H}_k)$ as $P_k = P_k^\oplus U_\mathcal{A}^\dagger$.\footnote{Note that this definition is consistent with the one introduced in the first paragraph above.} Hence, an arbitrary element $X_\mathcal{A} \in \mathcal{A}$ can be written as SOT-convergent series $X_\mathcal{A} = \sum_{i \in I} P_i^\dagger (X_{\mathtt{A}_i} \otimes \idmat_{\mathtt{B}_i}) P_i$, for some operators $X_{\mathtt{A}_i} \in \blt(\mathcal{H}_{\mathtt{A}_i})$, with $\sup_{i \in I} \norm{X_{\mathtt{A}_i}} < \infty$. 
\end{definition}
For $\mathcal{A} \subseteq \blt(\mathcal{H})$, we denote by $\mathcal{A}^\prime :=  \{X \in \blt(\mathcal{H}) \,|\, X X_\mathcal{A} = X_\mathcal{A} X, \,\forall X_\mathcal{A} \in \mathcal{A} \}$ its \emph{commutant}. If $\mathcal{A}$ is an atomic weakly closed *-algebra with decomposition given above, then as a special case of general theory of direct integral decompositions of vN-algebras (see e.g.   \cite[Proposition 14.1.24, Theorem 11.2.16]{kadison1997fundamentals2})), $\mathcal{A}^\prime$ is given by
\begin{align} \label{eq:FormOfAtomicCommutant}
    \mathcal{A}^\prime = U_\mathcal{A} \left( \blt(\mathcal{H}_0) \oplus \bigoplus_{i \in I}  (\idmat_{\mathtt{A}_i} \otimes \blt(\mathcal{H}_{\mathtt{B}_i}))\right) U_{\mathcal{A}}^\dagger.
\end{align}
Hence, an arbitrary element $X_\mathcal{A}^\prime \in \mathcal{A}^\prime$ can be written as SOT-convergent series
$X_{\mathcal{A}^\prime} = P_0^\dagger X_0 P_0 + \sum_{i \in I} P_i^\dagger (\idmat_{\mathtt{A}_i} \otimes X_{\mathtt{B}_i}) P_i$, for some operators $X_0 \in \blt(\mathcal{H}_0)$ and $X_{\mathtt{B}_i} \in \blt(\mathcal{H}_{\mathtt{B}_i})$, with $\sup_{i \in I} \norm{X_{\mathtt{B}_i}} < \infty$. And $X_{\mathcal{A}^\prime}$ is self-adjoint if and only if all the operators in the decomposition are self-adjoint \cite[Proposition 14.1.8]{kadison1997fundamentals2}.

\section{Results}

\subsection{GKLS-generators with invariant *-algebra} \label{GKLS-InvariantSection}

In this section we state and prove our first main result, namely a theorem that allows us to reduce the problem of characterizing GKLS-generators with invariant weakly closed *-algebras to characterizing CP-maps with invariant weakly closed *-algebras. Since CP-maps are special GKLS-generators (for $K = 0$), this renders these problems essentially equivalent. 
For technical reasons, we need to restrict ourselves to AFD algebras. The notation in the following theorem and the subsequent proof follows Section \ref{sec:NotationW*}.

\begin{theorem}{Theorem}\label{theorem:gkls-invariant-afd-algebra}
Let $L : \blt(\mathcal{H}) \rightarrow \blt(\mathcal{H})$ be defined by $L(X) = \Phi(X) - K^\dagger X - X K$, for some $\Phi \in \CP(\mathcal{H})$ and $K \in \blt(\mathcal{H})$, and let $\mathcal{A} \subseteq \blt(\mathcal{H})$ be an AFD weakly closed *-algebra. The following are equivalent

\begin{enumerate}
    \item $L(\mathcal{A}) \subseteq \mathcal{A}$. \label{cond:InvariantSubalgebra}
    \item (Stinespring) \label{cond:InvariantSubalgebraStinespring} Suppose $\Phi$ is given in Stinespring representation $\Phi(X) = V^\dagger (X \otimes \idmat_\mathtt{E}) V$, where $V \in \blt(\mathcal{H}; \mathcal{H} \otimes \mathcal{H}_\mathtt{E})$. Then there exist operators $V_0 \in \blt(\mathcal{H}; \mathcal{H}_0 \otimes \mathcal{H}_\mathtt{E})$, $A, B \in \blt(\mathcal{H}; \mathcal{H} \otimes \mathcal{H}_\mathtt{E})$ and $K_0 \in \blt(\mathcal{H}; \mathcal{H}_0)$; an operator $K_\mathcal{A} \in \mathcal{A}$; and a self-adjoint operator $H_{\mathcal{A}^\prime} \in \mathcal{A}^\prime$ such that 
    \begin{enumerate}
        \item \label{it:StinesprinConditionA}$A^\dagger (X_\mathcal{A} \otimes \idmat_\mathtt{E}) A \in \mathcal{A}$ and $(X_\mathcal{A} \otimes \idmat_\mathtt{E}) B = BX_\mathcal{A}$, for all $X_\mathcal{A} \in \mathcal{A}$.
        \item $V$ and $K$ have the following form:
        \begin{subequations}
        \begin{align}
        V &= (P_0^\dagger \otimes \idmat_\mathtt{E}) V_0 + A + B, \label{eq:VStinesprinForm}\\
        K &= B^\dagger A + \frac{1}{2} B^\dagger B + K_\mathcal{A} + i H_{\mathcal{A}^\prime} + P_0^\dagger K_0.\label{eq:KStinesprinForm}
        \end{align}
        \end{subequations}
    \end{enumerate}
    \item (Kraus) \label{cond:InvariantSubalgebraKraus} Suppose $\Phi$ is given in Kraus representation $\Phi(X) = \sum_{n \in N} \phi_n^\dagger X \phi_n$. Then there exists a countable index set $N$; collections of operators $\{v_n\}_{n \in N} \subset \blt(\mathcal{H}; \mathcal{H}_0)$ and $\{a_n\}_{n \in N}, \{b_n\}_{n \in N} \subset \blt(\mathcal{H})$ such that $\sum_{n\in N} v_n^\dagger v_n$, $\sum_{n\in N} a_n^\dagger a_n$ and $\sum_{n\in N} b_j^\dagger b_j$ SOT-converge; an operator $K_0 \in \blt(\mathcal{H}; \mathcal{H}_0)$; an operator $K_\mathcal{A} \in \mathcal{A}$; and a self-adjoint operator $H_{\mathcal{A}^\prime} \in \mathcal{A}^\prime$ such that 
    \begin{enumerate}
        \item $\sum_{n \in N} a_n^\dagger X_\mathcal{A} a_n \in \mathcal{A}$ for all $X_\mathcal{A} \in \mathcal{A}$ and $b_n \in \mathcal{A}^\prime$ for all $n \in N$. 
        \item $\{\phi_n\}_{n \in N}$ and $K$ have the following form:
        \begin{subequations}
        \begin{align}
            \phi_n &= P_0^\dagger v_n + a_n + b_n, \text{ for all } n \in N,\\
            K &= \sum_{n \in N} b_n^\dagger a_n + \frac{1}{2} \sum_{n \in N} b_n^\dagger b_n + K_\mathcal{A} + iH_{\mathcal{A}^\prime} + P_0^\dagger K_0.   
        \end{align}
        \end{subequations}
    \end{enumerate}
\end{enumerate}
\end{theorem}

\noindent \textbf{Remark.} If we take $\mathcal{X}_\mathcal{A} \in \mathcal{A}$, then a simple calculation (done in the first part of the proof) shows that
\begin{align*}
    L(X_\mathcal{A}) = A^\dagger (X_\mathcal{A} \otimes \idmat_\mathtt{E}) A - K_\mathcal{A}^\dagger X_\mathcal{A} - X_\mathcal{A} K_\mathcal{A}.
\end{align*}
This is the generator of a quantum dynamical semigroup on $\mathcal{A}$ and has the general form found by Christensen and Evans \cite[Theorem 3.1]{christensen1979cohomology}. Since the operators $A$ and $K_\mathcal{A}$ specify a generator on the subalgebra, the operators $B$, $K_0$, and $H_{\mathcal{A}^\prime}$ determine a specific one among the possible extensions of such a generator to all of $\blt(\mathcal{H})$. Note that if $\mathcal{A}$ is a von Neumann algebra, then the form simplifies since the term $P_0^\dagger K_0$ does not appear.  

\begin{proof}
We prove \ref{cond:InvariantSubalgebra} $\iff$ \ref{cond:InvariantSubalgebraStinespring} and obtain \ref{cond:InvariantSubalgebraKraus} as a corollary. For the implication \ref{cond:InvariantSubalgebraStinespring} $\implies$ \ref{cond:InvariantSubalgebra}, let $X_\mathcal{A} \in \mathcal{A}$ be arbitrary. We have
\begin{align*}
    \Phi(X_\mathcal{A}) &=\quad~ &&(V_0^\dagger (P_0 \otimes \idmat_\mathtt{E}) + A^\dagger + B^\dagger) (X_\mathcal{A} \otimes \idmat_\mathtt{E}) ((P_0^\dagger \otimes \idmat_\mathtt{E})V_0 + A + B) \\
    &\stackrel{\mathmakebox[\widthof{=}]{P_0X_\mathcal{A} = 0 = X_\mathcal{A}P_0^\dagger}}{=} &&(A^\dagger + B^\dagger) (X_\mathcal{A} \otimes \idmat_\mathtt{E}) (A + B) \\
    &\stackrel{\mathmakebox[\widthof{=}]{(X_\mathcal{A} \otimes \idmat_\mathtt{E})B = BX_\mathcal{A}}}{=} &&A^\dagger (X_\mathcal{A} \otimes \idmat_\mathtt{E}) A + (B^\dagger A + \frac{1}{2} B^\dagger B)^\dagger X_\mathcal{A}  + X_\mathcal{A} (B^\dagger A + \frac{1}{2} B^\dagger B)
\end{align*}
and
\begin{align*}
    K^\dagger X_\mathcal{A} + X_\mathcal{A}K &= (B^\dagger A + \frac{1}{2} B^\dagger B)^\dagger X_\mathcal{A}  + X_\mathcal{A} (B^\dagger A + \frac{1}{2} B^\dagger B) + K_\mathcal{A}^\dagger X_\mathcal{A} + X_\mathcal{A} K_\mathcal{A} \\
    &\hphantom{=}~- \underbrace{iH_{\mathcal{A}^\prime}X_\mathcal{A} + iX_\mathcal{A} H_{\mathcal{A}^\prime}}_{ = 0, \text{ since } H_{\mathcal{A}^\prime} \in \mathcal{A}^\prime } + \underbrace{K_0^\dagger P_0 X_\mathcal{A} + X_\mathcal{A}P_0^\dagger K_0}_{ = 0, \text{ since } P_0X_\mathcal{A} = 0 = X_\mathcal{A} P_0^\dagger}.
\end{align*}
Combining the calculations above yields
\begin{align*}
    L(X_\mathcal{A}) = A^\dagger (X_\mathcal{A} \otimes \idmat_\mathtt{E}) A - K_\mathcal{A}^\dagger X_\mathcal{A} - X_\mathcal{A} K_\mathcal{A},
\end{align*}
which belongs to $\mathcal{A}$ since by assumption $A^\dagger (X_\mathcal{A} \otimes \idmat_E) A \in \mathcal{A}$, $K_\mathcal{A}\in \mathcal{A}$ and $K_\mathcal{A}^\dagger \in \mathcal{A}$. \\
The proof of the converse proceeds in two main steps: First we show that there are operators $V_0$, $A$ and $B$ such that the conditions in \ref{it:StinesprinConditionA} and Eq.~\eqref{eq:VStinesprinForm} hold. Second, we derive the form of $K$. 
As a first step, we construct a family of linear maps on $\blt(\mathcal{H})$ each of which is closely related to $L$ and leaves $\mathcal{A}$ invariant.
Since $L(\mathcal{A}) \subseteq \mathcal{A}$, and since $\mathcal{A}$ is a *-algebra, 
\begin{align}
    \Psi(X, Y, Z) := L(Y^\dagger X Z) - Y^\dagger L(XZ) - L(Y^\dagger X)Z + Y^\dagger L(X)Z
\end{align}
is an element of $\mathcal{A}$ whenever $X, Y, Z \in \mathcal{A}$. A direct calculation using the representation $\Phi(X) = V^\dagger (X \otimes \idmat_{\mathtt{E}})V$ reveals that
\begin{align}
    \Psi(X, Y, Z) = \left[VY - (Y \otimes \idmat_E)V\right]^\dagger (X \otimes \idmat_E) \left[(VZ - (Z \otimes \idmat_E)V\right].
\end{align}
With the notation introduced in Section \ref{sec:NotationW*}: Since $\mathcal{A}$ is AFD, so is $U_\mathcal{A}^\dagger \mathcal{A} U_\mathcal{A} = 0_0 \oplus \mathcal{A}_{\overline{0}}$ and so is the vN-algebra $\mathcal{A}_{\overline{0}} \subseteq \blt(\mathcal{H}_{\overline{0}})$. Let $\tilde{\mathcal{A}}_1 \subseteq \tilde{\mathcal{A}}_2 \subseteq \tilde{\mathcal{A}}_3 \subseteq \cdots$ be an increasing sequence of finite-dimensional *-subalgebras of $\mathcal{A}_{\overline{0}}$, such that $\cup_{n \in \mathbb{N}} \tilde{\mathcal{A}}_n$ is WOT-dense in $\mathcal{A}_{\overline{0}}$. For every $n \in \mathbb{N}$, define $\mathcal{A}_n := \mathrm{span}\{\tilde{\mathcal{A}}_n \cup \mathbb{C} \idmat_{\overline{0}}\}$. Clearly, also $\cup_{n \in \mathbb{N}} \mathcal{A}_n$ is WOT-dense in $\mathcal{A}_{\overline{0}}$, but now $\mathcal{A}_n$ is a vN-algebra for every $n\in\mathbb{N}$. 
In the following, we will often need to assign to operators in $\mathcal{A}_{\overline{0}}$ the corresponding ones in $\mathcal{A}$. For notational convenience, we define for each $X \in \mathcal{A}_{\overline{0}}$ the operator $\hat{X} = P_{\overline{0}}^\dagger X P_{\overline{0}} \in \blt(\mathcal{H})$. We denote by $\mathcal{U}(\mathcal{A}_n)$ the unitary group in $\mathcal{A}_n$. 
As $\mathcal{A}_n$ is finite-dimensional, $\mathcal{U}(\mathcal{A}_n)$ is a compact group, so there exists a unique Haar probability measure on $\mathcal{U}(\mathcal{A}_n)$. For any $n, m \in \mathbb{N}$ and $X \in \blt(\mathcal{H})$, we obtain the following Haar average
\begin{align} \label{eq:HaarAverage1}
\begin{split}
    \int_{\unitary(\mathcal{A}_m)} &\int_{\unitary(\mathcal{A}_n)} \Psi(\hat{U}_n X \hat{W}_m^\dagger, \hat{U}_n^\dagger, \hat{W}_m) \, \mathrm{d}U_n \mathrm{d}W_m \\&= \left((\hat{\idmat}_{\overline{0}} \otimes \idmat_\mathtt{E})V -  \mathbb{E}_n(V)\right)^\dagger (X \otimes \idmat_\mathtt{E}) \left((\hat{\idmat}_{\overline{0}} \otimes \idmat_\mathtt{E})V - \mathbb{E}_m(V) \right),
\end{split}
\end{align}
where 
\begin{align*}
    \mathbb{E}_k(V) := \int_{\unitary(\mathcal{A}_k)} (\hat{U}_k^\dagger \otimes \idmat_\mathtt{E}) V \hat{U}_k \,\mathrm{d}U_k, \quad k \in \mathbb{N}.
\end{align*}
Since we integrate over a probability measure, $\norm{\mathbb{E}_k(V)} \leq \norm{V}$ and hence the (sequential) Banach–Alaoglu theorem implies that the sequence $(\mathbb{E}_k(V))_{k \in \mathbb{N}}$ has an ultraweakly convergent subsequence whose limit we denote by $\mathbb{E}(V)$. For $X_\mathcal{A} \in \mathcal{A}$, the RHS of Eq.~\eqref{eq:HaarAverage1} is an element of $\mathcal{A}$ for all $n, m \in \mathbb{N}$, since the integrand is in $\mathcal{A}$ and the Bochner integral converges in norm. Furthermore, since $\mathcal{A}$ is ultraweakly closed, passing to subsequences and taking the limit $n \rightarrow \infty$ and then $m \rightarrow \infty$ yields that 
\begin{align} \label{eq:CpInvariant}
    \Psi(X_\mathcal{A}) := \left((\hat{\idmat}_{\overline{0}} \otimes \idmat_\mathtt{E})V -  \mathbb{E}(V)\right)^\dagger (X_\mathcal{A} \otimes \idmat_\mathtt{E}) \left((\hat{\idmat}_{\overline{0}} \otimes \idmat_\mathtt{E})V - \mathbb{E}(V) \right)
\end{align}
is an element of $\mathcal{A}$ for all $X_\mathcal{A} \in \mathcal{A}$. In other words, $\Psi$ interpreted as a CP-map satisfies $\Psi(\mathcal{A}) \subseteq \mathcal{A}$. We now define $V_0$, $B$ and $A$ as follows: $V_0 = (P_0 \otimes \idmat_\mathtt{E}) V$, $B = \mathbb{E}(V)$ and $A = V - (P_0^\dagger \otimes \idmat_\mathtt{E})V_0 - B$. Thus $V = A + B + (P_0^\dagger \otimes \idmat_\mathtt{E})V$, which is precisely Eq.~\eqref{eq:VStinesprinForm}. It follows directly from Eq.~\eqref{eq:CpInvariant} that $A^\dagger (X_\mathcal{A} \otimes \idmat_\mathtt{E}) A \in \mathcal{A}$ for all $X_\mathcal{A} \in \mathcal{A}$ --- verifying the first part of condition \ref{it:StinesprinConditionA}. 
By the definition of the Haar measure and since $(\mathcal{A}_k)_{k \in \mathbb{N}}$ is an increasing sequence, we have $(\hat{U}_k \otimes \idmat_\mathtt{E}) \mathbb{E}(V) = \mathbb{E}(V)\hat{U}_k$ for all $U_k \in \mathcal{U}(\mathcal{A}_k)$. But since every $X_k \in \mathcal{A}_k$ can be written as a finite linear combination of elements in $\mathcal{U}(\mathcal{A}_k)$ (see \cite[Theorem 4.1.7]{kadison1997fundamentals}), we have $(\hat{X}_k \otimes \idmat_E) \mathbb{E}(V) = \mathbb{E}(V)\hat{X}_k$, for all $X_k \in \mathcal{A}_k$ and hence $(\hat{X} \otimes \idmat_\mathtt{E}) \mathbb{E}(V) = \mathbb{E}(V)\hat{X}$ for all $X \in \cup_{n \in \mathbb{N}} \mathcal{A}_n$. Evidently, this equation is also preserved under ultraweak limits. Thus $(X_\mathcal{A} \otimes \idmat_\mathtt{E}) \mathbb{E}(V) = \mathbb{E}(V)X_\mathcal{A}$, for all $X_\mathcal{A} \in \mathcal{A}$. Since $B = \mathbb{E}(V)$, this implies the second part of condition \ref{it:StinesprinConditionA}. 

It remains to show that $K$ has the desired form. To this end, note that for any $X_\mathcal{A} \in \mathcal{A}$, we have $L(X_\mathcal{A}) \in \mathcal{A}$ by assumption, but since $V = (P_0^\dagger \otimes \idmat_\mathtt{E}) V_0 + A + B$, $(X_\mathcal{A} \otimes \idmat_\mathtt{E})B = B X_\mathcal{A}$ and $P_0X_\mathcal{A} = 0$, we also have
\begin{align*}
    L(X_\mathcal{A}) &= \left[A + B\right]^\dagger (X_\mathcal{A} \otimes \idmat_\mathtt{E}) \left[A + B\right] - K^\dagger X_\mathcal{A} - X_\mathcal{A}K \\
    &= A^\dagger (X_\mathcal{A} \otimes \idmat_\mathtt{E}) A - (K -  B^\dagger A - \frac{1}{2} B^\dagger B )^\dagger X_\mathcal{A} - X_\mathcal{A} (K -  B^\dagger A - \frac{1}{2} B^\dagger B ).
\end{align*}
Since $A^\dagger (X_\mathcal{A} \otimes \idmat_\mathtt{E}) A \in \mathcal{A}$, this implies that
\begin{align} \label{eq:KappaInvarinatEquation}
    -\kappa^\dagger X_\mathcal{A} - X_\mathcal{A}\kappa \in \mathcal{A},  
\end{align}
for all $X_\mathcal{A} \in  \mathcal{A}$, where $\kappa = K - B^\dagger A - \frac{1}{2} B^\dagger B$. 
For $U_n \in \mathcal{U}(\mathcal{A}_n)$, we choose $X_\mathcal{A} = \hat{U}_n$, multiply Eq.~\eqref{eq:KappaInvarinatEquation} from the left by $\hat{U}_n^\dagger$ and integrate over the Haar measure. Thus, we see that
\begin{align} \label{eq:UltraweakKappaLimit}
    - \int_{\unitary(\mathcal{A}_n)} \hat{U}_n^\dagger \kappa \hat{U}_n \,\mathrm{d}U_n - \hat{\idmat}_{\overline{0}} \kappa
\end{align}
belongs to $\mathcal{A}$. By the same arguments as above, we can pass to a subsequence such that for $n \rightarrow \infty$, expression Eq.~\eqref{eq:UltraweakKappaLimit} converges to 
\begin{align} \label{eq:KappaAfterLinitEquation}
    - \kappa_{\mathcal{A}^\prime} - \hat{\idmat}_{\overline{0}} \kappa,
\end{align}
for some $\kappa_{\mathcal{A}^\prime} \in \mathcal{A}^\prime$ and such that the whole expression belongs to $\mathcal{A}$. We now define the self-adjoint operator $H_{\mathcal{A}^\prime} = - \frac{1}{2i} (\kappa_{\mathcal{A}^\prime} - \kappa_{\mathcal{A}^\prime}^\dagger) \in \mathcal{A}^\prime$, the operator $K_\mathcal{A} :=  \hat{\idmat}_{\overline{0}} \kappa - iH_{\mathcal{A}^\prime} = \kappa - P_0^\dagger P_0\kappa - iH_{\mathcal{A}^\prime}$ and $K_0 = P_0\kappa$. By the definition of $K_0$ and $K_\mathcal{A}$ we thus get $\kappa = K_\mathcal{A} + iH_{\mathcal{A}^\prime} + P_0^\dagger K_0$, which is the desired form if $K_\mathcal{A} \in \mathcal{A}$. This last assertion can be seen as follows:
\begin{align*}
    K_\mathcal{A} &= \frac{1}{2}\left(K_\mathcal{A} + K_\mathcal{A}^\dagger \right) + \frac{1}{2} \left(K_\mathcal{A} - K_\mathcal{A}^\dagger\right) \\
    &= \underbrace{\frac{1}{2} \left(\kappa^\dagger \hat{\idmat}_{\overline{0}} + \hat{\idmat}_{\overline{0}} \kappa \right)}_{\in \mathcal{A}, \text{ by }  \hat{\idmat}_{\overline{0}} \in \mathcal{A} \text{ and Eq.~} \eqref{eq:KappaInvarinatEquation}}  + \frac{1}{2} \bigg( \underbrace{(\hat{\idmat}_{\overline{0}} \kappa + \kappa_{\mathcal{A}^\prime})}_{\in\mathcal{A}, \text{ by Eq.~} \eqref{eq:KappaAfterLinitEquation}} - \underbrace{(\hat{\idmat}_{\overline{0}} \kappa + \kappa_{\mathcal{A}^\prime})^\dagger}_{\in\mathcal{A}, \text{ by Eq.~} \eqref{eq:KappaAfterLinitEquation}} \bigg).
\end{align*}
This finishes the proof of \ref{cond:InvariantSubalgebra} $\iff$ \ref{cond:InvariantSubalgebraStinespring}. 

Part \ref{cond:InvariantSubalgebraKraus} is a matter of going from the Stinespring representation of normal CP-maps to their Kraus representation and back. This is a standard procedure and a very nice account can be found in \cite{AttalCPMaps}. We just mention here that after choosing an orthonormal basis $\{\ket{e_n} \}_{n \in N}$ of $\mathcal{H}_\mathtt{E}$, the collections $\{v_n\}_{n \in N}$, $\{a_n\}_{n \in N}$ and $\{b_n\}_{n \in N}$ and the operators $V_0$, $A$ and $B$ are related via $v_n := (\idmat \otimes \bra{e_n})V_0$, $a_n := (\idmat \otimes \bra{e_n})A$ and $b_n := (\idmat \otimes \bra{e_n})B$. The corresponding properties are then routinely verifiable.  
\end{proof}

\subsection{CP-maps with invariant atomic algebra} \label{CP-invarinatSection}

In this section, we study the problem of finding a normal form for (normal) CP-maps with an \textit{atomic} invariant subalgebra. Slightly more generally, we aim to find normal a normal form for normal CP-maps $\Phi$ with the property that $\Phi(\mathcal{A}) \subseteq \mathcal{C}$, for two atomic weakly closed *-algebras $\mathcal{A}$ and $\mathcal{C}$. Since we are now dealing with two algebras, we need to distinguish them in the notation in Definition \ref{Defn:NormalFormAtomicAlgebra}. For the algebra $\mathcal{A} \subseteq \blt(\mathcal{H}_\mathcal{A})$: the index set is called $I$; the Hilbert spaces $\{\mathcal{H}_i\}_{i \in I \cup \{0\}}$ are denoted by $\mathcal{H}_{i:\mathcal{A}}$, with $\mathcal{H}_{i:\mathcal{A}} = \mathcal{H}_{\mathtt{A}_i} \otimes \mathcal{H}_{\mathtt{B}_i}$ ($i \in I$); and the operators $P_i$ are called $P_{i:\mathcal{A}} \in \blt(\mathcal{H}_\mathcal{A}; \mathcal{H}_{i:\mathcal{A}})$.  
For the algebra $\mathcal{C} \subseteq \blt(\mathcal{H}_\mathcal{C})$: the index set is called $J$; the Hilbert spaces $\{\mathcal{H}_j\}_{j \in J \cup \{0\}}$ are denoted by $\mathcal{H}_{j:\mathcal{C}}$, with $\mathcal{H}_{j:\mathcal{C}} = \mathcal{H}_{\mathtt{C}_j} \otimes \mathcal{H}_{\mathtt{D}_j}$ ($j \in J$); and the operators $P_j$ are called $P_{j:\mathcal{C}} \in \blt(\mathcal{H}_\mathcal{C}; \mathcal{H}_{j:\mathcal{C}})$. 
With this notation in place, we can state our second main result:

\begin{theorem}{Theorem}\label{theorem:cp-map-algebra-to-algebra-stinespring}
   Let $\mathcal{A} \subseteq \blt(\mathcal{H}_{\mathcal{A}})$ and $\mathcal{C} \subseteq \blt(\mathcal{H}_{\mathcal{C}})$ be two atomic weakly closed *-algebras. 
   For $\Phi \in \CP(\mathcal{H}_\mathcal{A}; \mathcal{H}_\mathcal{C})$ defined by $\Phi(X) = V^\dagger (X \otimes \idmat_\mathtt{E}) V$, with $V \in \blt(\mathcal{H}_\mathcal{C};\mathcal{H}_\mathcal{A} \otimes \mathcal{H}_\mathtt{E})$, the following are equivalent
   \begin{enumerate}
       \item \label{it:AToCCondition}$\Phi(\mathcal{A}) \subseteq \mathcal{C}$.
       \item \label{CharacterizationConditionCPTheorem} There exist an operator $V_0 \in \blt(\mathcal{H}_\mathcal{C}; \mathcal{H}_{0:\mathcal{A}} \otimes \mathcal{H}_\mathtt{E})$; and for all $i \in I$ and $j \in J$ 
        Hilbert spaces $\mathcal{H}_{\mathtt{F}_{ij}}$, operators $A_{ij} \in \blt(\mathcal{H}_{\mathtt{C}_j}; \mathcal{H}_{\mathtt{A}_i} \otimes \mathcal{H}_{\mathtt{F}_{ij}})$, and isometries $U_{ij} \in \blt(\mathcal{H}_{\mathtt{F}_{ij}} \otimes \mathcal{H}_{\mathtt{D}_j}; \mathcal{H}_{\mathtt{B}_i} \otimes \mathcal{H}_\mathtt{E})$, such that
        \begin{itemize}
            \item $V$ can be decomposed as
            \begin{align} \label{eq:normalFormCPV}
            V = (P_{0:\mathcal{A}}^\dagger \otimes \idmat_\mathtt{E}) V_0 + \sum_{i \in I,\, j\in J} (P_{i: \mathcal{A}}^\dagger \otimes \idmat_\mathtt{E}) V_{ij} P_{j:\mathcal{C}}, 
            \end{align}
            with $V_{ij} = (\idmat_{\mathtt{A}_i} \otimes U_{ij})(A_{ij} \otimes \idmat_{\mathtt{D}_j})$, s.t. the series SOT-converges.
        \item The relation $U_{ik}^\dagger U_{il} = \delta_{kl} \idmat$ holds for all $i \in I$ and $k, l \in J$. 
        \end{itemize}
   \end{enumerate}
   The representation in \ref{CharacterizationConditionCPTheorem} can be chosen such that $\{(X \otimes \idmat_{\mathtt{F}_{ij}})A_{ij}\ket{\psi} \, | \, X \in \blt(\mathcal{H}_{\mathtt{A}_i}), \ket{\psi} \in \mathcal{H}_{\mathtt{C}_j} \}$ is total in $\mathcal{H}_{\mathtt{A}_i} \otimes \mathcal{H}_{\mathtt{F}_{ij}}$.
\end{theorem}

\noindent \textbf{Remark.} The theorem above tells us that if $V$ is written as a block matrix (w.r.t. a basis determined by the structure of $\mathcal{A}$ and $\mathcal{C}$), then all the blocks are necessarily of the ``semilocalizable'' form $(\idmat \otimes U)(A \otimes \idmat)$ --- and that the $U$'s need to satisfy an orthogonality relation. 

\begin{proof}
Let us start by showing that \ref{CharacterizationConditionCPTheorem} $\implies$ \ref{it:AToCCondition}. We know that $\mathcal{X}_\mathcal{A} \in \mathcal{A}$ if and only if it can be decomposed as SOT-convergent series $\mathcal{X}_\mathcal{A} = \sum_{i \in I} P_{i:\mathcal{A}}^\dagger (X_{\mathtt{A}_i} \otimes \idmat_{\mathtt{B}_i}) P_{i:\mathcal{A}}$, for $X_{\mathtt{A}_i} \in \blt(\mathcal{H}_{\mathtt{A}_i})$ with $\sum_{i\in I} \norm{X_{\mathtt{A}_i}} < \infty$. Since $\mathcal{A}$ and $\mathcal{C}$ are ultraweakly closed, $\Phi$ is ultraweakly continuous, and the operators in $\mathcal{A}_F := \{ X_\mathcal{A} \,|\, X_{\mathtt{A}_i} \neq 0 \text{ only for finitely many $i \in I$ }\}$ are ultraweakly dense in $\mathcal{A}$, it suffices to show the claim for $X_\mathcal{A} \in \mathcal{A}_F$ so that convergence issues (w.r.t.~the $I$-summation) play no role in the following calculation (where the $J$-summation SOT-converges):
\begin{align*}
    \Phi(X_\mathcal{A}) &= \sum_{i \in I}\left[(P_{i:\mathcal{A}} \otimes \idmat_\mathtt{E})V\right]^\dagger (X_{\mathtt{A}_i} \otimes \idmat_{\mathtt{B}_i} \otimes \idmat_\mathtt{E})\left[(P_{i:\mathcal{A}} \otimes \idmat_\mathtt{E}) V\right] \\
    &= \sum_{i \in I}\sum_{k,l \in J}  P_{k:\mathcal{C}}^\dagger V_{ik}^\dagger (X_{\mathtt{A}_i} \otimes \idmat_{\mathtt{B}_i} \otimes \idmat_\mathtt{E}) V_{il} P_{l:\mathcal{C}} \\
    &= \sum_{i \in I}\sum_{k,l \in J}  P_{k:\mathcal{C}}^\dagger (A_{ik}^\dagger \otimes \idmat_{\mathtt{D}_k}) (X_{\mathtt{A}_i} \otimes U_{ik}^\dagger U_{il})(A_{il} \otimes \idmat_{\mathtt{D}_l}) P_{l:\mathcal{C}} \\
    &= \sum_{i \in I}\sum_{j \in J} P_{j:\mathcal{C}}^\dagger (A_{ij}^\dagger \otimes \idmat_{\mathtt{D}_j}) (X_{\mathtt{A}_i} \otimes \idmat_{\mathtt{F}_{ij}} \otimes \idmat_{\mathtt{D}_j})(A_{ij} \otimes \idmat_{\mathtt{D}_j}) P_{j:\mathcal{C}} \\
    &= \sum_{j \in J} P_{j:\mathcal{C}}^\dagger \left[\left(\sum_{i \in I} A_{ij}^\dagger (X_{\mathtt{A}_i} \otimes \idmat_{\mathtt{F}_{ij}}) A_{ij} \right) \otimes \idmat_{\mathtt{D}_j} \right] P_{j:\mathcal{C}},
\end{align*}
where we used the expansion of $X_\mathcal{A}$ in the first line, Eq.~\eqref{eq:normalFormCPV} in the second line (in particular the orthogonality of the projections), the explicit form $V_{ij} = (\idmat_{\mathtt{A}_i} \otimes U_{ij})(A_{ij} \otimes \idmat_{\mathtt{D}_j})$ in the third line, the orthogonality relation $U_{ik}^\dagger U_{il} = \delta_{kl}\idmat$ in the fourth line and algebraic manipulations in the fifth line. But the last line is just the decomposed form of an element of $\mathcal{C}$. Thus we have shown that $\Phi(\mathcal{A}) \subseteq \mathcal{C}$. 

For the converse, suppose that $\Phi$, defined by $\Phi(X) = V^\dagger (X \otimes \idmat_\mathtt{E}) V$ satisfies $\Phi(\mathcal{A}) \subseteq \mathcal{C}$. Let $I_0 = I \cup \{0\}$ and $J_0 = J \cup \{0\}$. Then $\sum_{i_0 \in I_0} P_{i_0:\mathcal{A}}^\dagger P_{i_0:\mathcal{A}} = \idmat_\mathcal{A}$ and $\sum_{j_0 \in J_0} P_{j_0:\mathcal{C}}^\dagger P_{j_0:\mathcal{C}} = \idmat_\mathcal{C}$, where the series SOT-converge. Hence, we can expand
\begin{align} \label{eq:VExpansionNaive}
    V = (P_{0: \mathcal{A}}^\dagger \otimes \idmat_\mathtt{E}) V_0 + \sum_{i \in I} (P_{i:\mathcal{A}}^\dagger \otimes \idmat_\mathtt{E}) V_{i0} P_{0:\mathcal{C}} + \sum_{i \in I, j \in J} (P_{i:\mathcal{A}}^\dagger \otimes \idmat_\mathtt{E}) V_{ij} P_{j:\mathcal{C}},
\end{align}
where we defined $V_0 = (P_{0:\mathcal{A}} \otimes \idmat_\mathtt{E})V$ and for all $i \in I$ and $j_0 \in J_0$ the operator $V_{ij_0} = (P_{i:\mathcal{A}} \otimes \idmat_\mathtt{E}) V P_{j_0:\mathcal{C}}^\dagger$. Thus it remains to show that $V_{i0} = 0$ and that $V_{ij}$ has our specific form. 

By definition, every $X_\mathcal{C} \in \mathcal{C}$ is of the form $X_\mathcal{C} = \sum_{j \in J} P_{j:\mathcal{C}}^\dagger (X_{\mathtt{C}_j} \otimes \idmat_{\mathtt{D}_j}) P_{j:\mathcal{C}}$. Thus, in particular $\Phi(X_\mathcal{A})$ assumes that form for all $X_\mathcal{A} \in \mathcal{A}$. This has the following three implications: 
First, for every $i \in I$, $j \in J$ and every $X_{\mathtt{A}_i} \in \blt(\mathcal{H}_{\mathtt{A}_i})$, there exists $X_{\mathtt{C}_j} \in \blt(\mathcal{H}_{\mathtt{C}_j})$ such that
\begin{align} \label{eq:FirstConsequenceInvariance}
\begin{split}
    P_{j:\mathcal{C}} \Phi\left(P_{i:\mathcal{A}}^\dagger (X_{\mathtt{A}_i} \otimes \idmat_{\mathtt{B}_i}) P_{i:\mathcal{A}}\right) P_{j:\mathcal{C}}^\dagger = V_{ij}^\dagger (X_{\mathtt{A}_i} \otimes \idmat_{\mathtt{B}_i} \otimes \idmat_\mathtt{E}) V_{ij}  = X_{\mathtt{C}_j} \otimes \idmat_{\mathtt{D}_j}. 
    \end{split}
\end{align}
Second, for every $i \in I$ and $X_{\mathtt{A}_i} \in \blt(\mathcal{H}_{\mathtt{A}_i})$, we have 
\begin{align} \label{eq:SecondConsequenceInvariance}
    P_{0:\mathcal{C}} \,\Phi\left(P_{i:\mathcal{A}}^\dagger (X_{\mathtt{A}_i} \otimes \idmat_{\mathtt{B}_i}) P_{i:\mathcal{A}}\right) \,P_{0:\mathcal{C}}^\dagger = V_{i0}^\dagger (X_{\mathtt{A}_i} \otimes \idmat_{\mathtt{B}_i} \otimes \idmat_\mathtt{E}) V_{i0} = 0. 
\end{align}
Third, for every $i \in I$ and $k,l \in J$ with $k \neq l$ and all $X_{\mathtt{A}_i} \in \blt(\mathcal{H}_{\mathtt{A}_i})$, we have
\begin{align} \label{eq:ThirdConsequenceInvariance}
    P_{k:\mathcal{C}} \,\Phi\left(P_{i:\mathcal{A}}^\dagger (X_{\mathtt{A}_i} \otimes \idmat_{\mathtt{B}_i}) P_{i:\mathcal{A}}\right) \,P_{l:\mathcal{C}}^\dagger = V_{ik}^\dagger (X_{\mathtt{A}_i} \otimes \idmat_{\mathtt{B}_i} \otimes \idmat_\mathtt{E}) V_{il} = 0. 
\end{align}

It is a direct consequence of Eq.~\eqref{eq:SecondConsequenceInvariance}, by choosing $X_{\mathtt{A}_i} = \idmat_{\mathtt{A}_i}$, that $V_{i0} = 0$ for all $i \in I$. Hence the second term in Eq.~\eqref{eq:VExpansionNaive} vanishes as desired.  

We will now see that Eq.~\eqref{eq:FirstConsequenceInvariance} implies that $V_{ij} = (\idmat \otimes U_{ij})(A_{ij} \otimes \idmat)$. This is due to the 
equivalence between semicausal and semilocalizable CP-maps, established for finite-dimensional systems in~\cite{Eggeling.2002}. We reproduce the argument here for the infinite-dimensional case. Choose some unit vector $\ket{\psi} \in \mathcal{H}_{\mathtt{D}_j}$ and define the normal CP-maps $\Phi_{ij} \in \CP(\mathcal{H}_{i:\mathcal{A}}; \mathcal{H}_{j:\mathcal{C}})$ and $\Psi_{ij} \in \CP(\mathcal{H}_{\mathtt{A}_i}; \mathcal{H}_{\mathtt{C}_j})$ by $\Phi_{ij}(X_i) = P_{j:\mathcal{C}} \Phi\left(P_{i:\mathcal{A}}^\dagger X_i P_{i:\mathcal{A}}\right) P_{j:\mathcal{C}}^\dagger$ and $\Psi_{ij}(X_{\mathtt{A}_i}) =  (\idmat_{\mathtt{C}_j} \otimes \bra{\psi})\Phi_{ij}(X_{\mathtt{A}_i} \otimes \idmat_{\mathtt{B}_i})(\idmat_{\mathtt{C}_j} \otimes \ket{\psi})$.
Eq.~\eqref{eq:FirstConsequenceInvariance} then implies that 
\begin{align} \label{eq:SemicausalityInProof}
    \Phi_{ij}(X_{\mathtt{A}_i} \otimes \idmat_{\mathtt{B}_i}) = \Psi_{ij}(X_{\mathtt{A}_i}) \otimes \idmat_{\mathtt{D}_j},
\end{align}
for all $X_{\mathtt{A}_i} \in \blt(\mathcal{H}_{\mathtt{A}_i})$. By Stinespring's dilation theorem (see Theorem \ref{thm:UniquenessStinespringThm}), there exists a minimal dilation given by $\mathcal{H}_{\mathtt{F}_{ij}}$ and $A_{ij} \in \blt(\mathcal{H}_{\mathtt{C}_j}; \mathcal{H}_{\mathtt{A}_i} \otimes \mathcal{H}_{\mathtt{F}_{ij}})$ such that $\Psi_{ij}(X_{\mathtt{A}_i}) = A_{ij}^\dagger (X_{\mathtt{A}_i} \otimes \idmat_{\mathtt{F}_{ij}}) A_{ij}$. It follows that $A_{ij} \otimes \idmat_{\mathtt{D}_{j}}$ is a minimal dilation for $X_{\mathtt{A}_i} \mapsto \Psi_{ij}(X_{\mathtt{A}_i}) \otimes \idmat_{\mathtt{D}_j}$. But Eqs.~\eqref{eq:SemicausalityInProof} and \eqref{eq:FirstConsequenceInvariance} then imply that
\begin{align}
    V_{ij}^\dagger (X_{\mathtt{A}_i} \otimes \idmat_{\mathtt{B}_i} \otimes \idmat_\mathtt{E}) V_{ij} = (A_{ij} \otimes \idmat_{\mathtt{D}_{j}})^\dagger (X_{\mathtt{A}_i} \otimes \idmat_{\mathtt{F}_{ij}} \otimes \idmat_{\mathtt{D}_j}) (A_{ij} \otimes \idmat_{\mathtt{D}_{j}}),
\end{align}
for all $X_{\mathtt{A}_i} \in \blt(\mathcal{H}_{\mathtt{A}_i})$. In other words, $V_{ij}$ and $A_{ij} \otimes \idmat_{\mathtt{D}_{j}}$ are Stinespring operators of the same map. Since $A_{ij} \otimes \idmat_{\mathtt{D}_{j}}$ is minimal, there exists an isometry $U_{ij} \in \blt(\mathcal{H}_{\mathtt{F}_{ij}} \otimes \mathcal{H}_{\mathtt{D}_j}; \mathcal{H}_{\mathtt{B}_i} \otimes \mathcal{H}_\mathtt{E})$ such that $V_{ij} = (\idmat_{\mathtt{A}_i} \otimes U_{ij})(A_{ij} \otimes \idmat_{\mathtt{D}_j})$. This is the desired form. 

It remains to show that $U_{ik}^\dagger U_{il} = \delta_{kl} \idmat$ for all $i \in I$ and $k, l \in J$. Since the $U_{ij}$'s are isometries, this condition is fulfilled for $k = l$. For $k \neq l$, we choose arbitrary $\ket{\psi_k} \in \mathcal{H}_{k:\mathcal{C}}$ and $\ket{\psi_l} \in \mathcal{H}_{l:\mathcal{C}}$ and $X_L, X_R \in \blt(\mathcal{H}_{\mathtt{A}_i})$. Eq.~\eqref{eq:ThirdConsequenceInvariance}, with $X_{\mathtt{A}_i} = X_L^\dagger X_R$ implies
\begin{align*}
    0 &= \braket{(X_L \otimes \idmat_{\mathtt{B}_i} \otimes \idmat_\mathtt{E})V_{ik} \psi_k}{(X_R \otimes \idmat_{\mathtt{B}_i} \otimes \idmat_\mathtt{E})V_{il} \psi_l} \\&= \braket{(X_L \otimes \idmat_{\mathtt{F}_{ik}} \otimes \idmat_{\mathtt{D}_k} )(A_{ik} \otimes \idmat_{\mathtt{D}_k}) \psi_k}{\left[\idmat_{\mathtt{A}_i} \otimes U_{ik}^\dagger U_{il} \right] (X_R \otimes \idmat_{\mathtt{F}_{il}} \otimes \idmat_{\mathtt{D}_l})(A_{il} \otimes \idmat_{\mathtt{D}_l}) \psi_l}.
\end{align*}
Since $\{(X_L \otimes \idmat_{\mathtt{F}_{ik}} \otimes \idmat_{\mathtt{D}_k} )(A_{ik} \otimes \idmat_{\mathtt{D}_k}) \ket{\psi_k} \, | \, X_L \in \blt(\mathcal{H}_{\mathtt{A}_i}), \ket{\psi_k} \in \mathcal{H}_{k:\mathcal{C}}\}$ being total and $\{(X_R \otimes \idmat_{\mathtt{F}_{il}} \otimes \idmat_{\mathtt{D}_l})(A_{il} \otimes \idmat_{\mathtt{D}_l}) \ket{\psi_l} \, | \, X_R \in \blt(\mathcal{H}_{\mathtt{A}_i}), \ket{\psi_l} \in \mathcal{H}_{l:\mathcal{C}}\}$ being total is the definition of minimality of $A_{ik} \otimes \idmat_{\mathtt{D}_k}$ and $A_{il} \otimes \idmat_{\mathtt{D}_l}$, respectively, we can conclude from the equation above (using sesquilinearity of the inner product) that $\idmat_{\mathtt{A}_i} \otimes U_{ik}^\dagger U_{il} = 0$ and hence that $U_{ik}^\dagger U_{il} = 0$, as desired. Finally, note that the claim about the totality of $\{(X \otimes \idmat_{\mathtt{F}_{ij}})A_{ij}\ket{\psi} \, | \, X \in \blt(\mathcal{H}_{\mathtt{A}_i}), \ket{\psi} \in \mathcal{H}_{\mathtt{C}_j} \}$ in $\mathcal{H}_{\mathtt{A}_i} \otimes \mathcal{H}_{\mathtt{F}_{ij}}$ follows by construction.
\end{proof}

\subsection{GKLS-generators with invariant atomic algebra} \label{AtomicInvarinatGKSL}

The notation in the following theorem and its proof follows Section \ref{sec:NotationW*}.

\begin{theorem}{Theorem} \label{THM:AtomicLindbladNormalForm}
   Let $L : \blt(\mathcal{H}) \rightarrow \blt(\mathcal{H})$ be given by $L(X) = V^\dagger (X \otimes \idmat_\mathtt{E}) V - K^\dagger X - X K$ with $V\in\blt (\mathcal{H};\mathcal{H}\otimes\mathcal{H}_{\mathtt{E}})$ and $K\in\blt (\mathcal{H})$ and let $\mathcal{A}$ be an atomic *-subalgebra of $\blt(\mathcal{H})$, with decomposition given by Definition \ref{Defn:NormalFormAtomicAlgebra}.  Then the following are equivalent
   \begin{enumerate}
       \item $L(\mathcal{A}) \subseteq \mathcal{A}$.
       \item \label{item:NormalFormAtomicLindblad} There exist operators $V_0 \in \blt(\mathcal{H}; \mathcal{H}_{0} \otimes \mathcal{H}_\mathtt{E})$ and $K_0 \in \blt(\mathcal{H}; \mathcal{H}_0)$;  for all $i, j \in I$
        a Hilbert space $\mathcal{H}_{\mathtt{F}_{ij}}$, operators $A_{ij} \in \blt(\mathcal{H}_{\mathtt{A}_j}; \mathcal{H}_{\mathtt{A}_i} \otimes \mathcal{H}_{\mathtt{F}_{ij}})$, and isometries $U_{ij} \in \blt(\mathcal{H}_{\mathtt{F}_{ij}} \otimes \mathcal{H}_{\mathtt{B}_j}; \mathcal{H}_{\mathtt{B}_i} \otimes \mathcal{H}_\mathtt{E})$; and for every $i \in I$ operators $B_i \in \blt(\mathcal{H}_{\mathtt{B}_i}; \mathcal{H}_{\mathtt{B}_i} \otimes \mathcal{H}_{\mathtt{E}})$, $K_{\mathtt{A}_i} \in \blt(\mathcal{H}_{\mathtt{A}_i})$, and self-adjoint operators $H_{\mathtt{B}_i} \in \blt(\mathcal{H}_{\mathtt{B}_i})$, such that 
        \begin{itemize}
            \item $V$ and $K$ can be decomposed as
            \begin{align*}
               V &= (P_0^\dagger \otimes \idmat_\mathtt{E})V_0 + \sum_{i, j \in I} (P_i^\dagger \otimes \idmat_\mathtt{E}) V_{ij}^{sc} P_j + \sum_{i \in I} (P_i^\dagger \otimes \idmat_\mathtt{E}) (\idmat_{\mathtt{A}_i} \otimes B_i) P_i,\\
               K &= 
               \sum_{i \in I} P_i^\dagger (\idmat_{\mathtt{A}_i} \otimes B_i^\dagger) V_{ii}^{sc} P_i + \frac{1}{2}\sum_{i \in I} P_i^\dagger (\idmat_{\mathtt{A}_i} \otimes B_i^\dagger B_i) P_i \\
               &\hphantom{=}~+ K_\mathcal{A} + iH_{\mathcal{A}^\prime} + P_0^\dagger K_0,
           \end{align*}
           with $V_{ij}^{sc} = (\idmat_{\mathtt{A}_i} \otimes U_{ij})(A_{ij} \otimes \idmat_{\mathtt{B}_j})$, $K_\mathcal{A} = \sum_{i \in I} P_i^\dagger (K_{\mathtt{A}_i} \otimes \idmat_{\mathtt{B}_i}) P_i$, and $H_{\mathcal{A}^\prime} = \sum_{i \in I} P_i^\dagger (\idmat_{\mathtt{A}_i} \otimes H_{\mathtt{B}_i}) P_i$, s.t. all series SOT-converge.  
            \item The relation $U_{ik}^\dagger U_{il} = \delta_{kl} \idmat$ holds for all $i, k, l \in I$.
        \end{itemize}
   \end{enumerate}
\end{theorem}
\begin{proof}
The basic strategy is to use Theorem \ref{theorem:gkls-invariant-afd-algebra} to reduce the problem to CP-maps with invariant algebra $\mathcal{A}$, followed by an application of Theorem \ref{theorem:cp-map-algebra-to-algebra-stinespring}. In detail: Part \ref{cond:InvariantSubalgebraStinespring} of Theorem \ref{theorem:gkls-invariant-afd-algebra} provides us with operators $\tilde{A}$, $\tilde{B}$, $\tilde{V}_0$, $\tilde{K}_0$, $K_\mathcal{A}$ and $\tilde{H}_{\mathcal{A}^\prime}$ such that
\begin{subequations}
\begin{align}
    V &= (P_0^\dagger \otimes \idmat_\mathtt{E})\tilde{V}_0 + \tilde{A} + \tilde{B}, \label{eq:VInProof}\\
    K &= \tilde{B}^\dagger \tilde{A} + \frac{1}{2} \tilde{B}^\dagger \tilde{B} + K_\mathcal{A} + i\tilde{H}_{\mathcal{A}^\prime} + P_0^\dagger \tilde{K}_0. \label{eq:KInProofObtainedFromGeneralInvarinatTheorem}
\end{align}
\end{subequations}
We observe the following:
\begin{itemize}
    \item Since $K_\mathcal{A} \in \mathcal{A}$, it can be decomposed as $K_\mathcal{A} = \sum_{i \in I} P_i^\dagger (K_{\mathtt{A}_i} \otimes \idmat_{\mathtt{B}_i}) P_i$ for operators $K_{\mathtt{A}_i} \in \blt(\mathcal{H}_{\mathtt{A}_i})$.
    \item Since $\Phi(X_\mathcal{A}) := \tilde{A}^\dagger (X_\mathcal{A} \otimes \idmat_\mathtt{E}) \tilde{A} \in \mathcal{A}$ for all $X_\mathcal{A} \in \mathcal{A}$, we can apply Theorem \ref{theorem:cp-map-algebra-to-algebra-stinespring}, which implies that there exist $A_0 \in \blt(\mathcal{H}; \mathcal{H}_0 \otimes \mathcal{H}_\mathtt{E})$; and for all $i, j \in I$, a Hilbert space $\mathcal{H}_{\mathtt{F}_{ij}}$, operators $A_{ij} \in \blt(\mathcal{H}_{\mathtt{A}_j}; \mathcal{H}_{\mathtt{A}_i} \otimes \mathcal{H}_{\mathtt{F}_{ij}})$, and isometries $U_{ij} \in \blt(\mathcal{H}_{\mathtt{F}_{ij}} \otimes \mathcal{H}_{\mathtt{B}_j}; \mathcal{H}_{\mathtt{B}_i} \otimes \mathcal{H}_\mathtt{E})$ such that $\tilde{A} = (P_0^\dagger \otimes \idmat_\mathtt{E})A_0 + \sum_{i,j \in I} (P_i^\dagger \otimes \idmat_\mathtt{E}) V_{ij}^{sc} P_j$, where $V_{ij}^{sc} = (\idmat_{\mathtt{A}_i} \otimes U_{ij})(A_{ij} \otimes \idmat_{\mathtt{B}_j})$ and $U_{ik}^\dagger U_{il} = \delta_{kl}\idmat$ for all $i,k,l \in I$.
    \item Since $\tilde{B}$ satisfies $(X_\mathcal{A} \otimes \idmat_{\mathtt{E}})\tilde{B} = \tilde{B} X_\mathcal{A}$ for all $X_\mathcal{A} \in \mathcal{A}$, a calculation executed in Lemma \ref{lem:HalfCommutingBsLemma} shows that there exist $B_0 \in \blt(\mathcal{H}_0; \mathcal{H}_0 \otimes \mathcal{H}_\mathtt{E})$ and operators $B_i \in \blt(\mathcal{H}_{\mathtt{B}_i}; \mathcal{H}_{\mathtt{B}_i} \otimes \mathcal{H}_\mathtt{E})$ such that $\tilde{B} = (P_0^\dagger \otimes \idmat_\mathtt{E}) B_0 P_0 + \sum_{i \in I} (P_i^\dagger \otimes \idmat_\mathtt{E}) (\idmat_{\mathtt{A}_i} \otimes B_i) P_i$.
    \item Since $\tilde{H}_{\mathcal{A}^\prime} \in \mathcal{A}^\prime$, the discussion around Eq.~\eqref{eq:FormOfAtomicCommutant} yields that it can be decomposed as $\tilde{H}_{\mathcal{A}^\prime} = P_0^\dagger H_0 P_0 + \sum_{i \in I} P_i^\dagger (\idmat_{\mathtt{A}_i} \otimes H_{\mathtt{B}_i}) P_i$, for self-adjoint $H_0 \in \blt(\mathcal{H}_0)$ and $H_{\mathtt{B}_i} \in \blt(\mathcal{H}_{\mathtt{B}_i})$. 
\end{itemize}
Each of the points above provides an explicit representation of the operators in Eqs.~\eqref{eq:VInProof} and \eqref{eq:KInProofObtainedFromGeneralInvarinatTheorem}. Plugging them in yields
\begin{align*}
    V &= (P_0^\dagger \otimes \idmat_\mathtt{E})\tilde{V}_0 + (P_0^\dagger \otimes \idmat_\mathtt{E})A_0 + (P_0^\dagger \otimes \idmat_\mathtt{E})B_0P_0 \\
    &\hphantom{=}~+ \sum_{i, j \in I} (P_i^\dagger \otimes \idmat_\mathtt{E}) V_{ij}^{sc} P_j + \sum_{i \in I} (P_i^\dagger \otimes \idmat_\mathtt{E}) (\idmat_{A_i} \otimes B_i) P_i,
\end{align*}
which has the desired form after defining $V_0 = \tilde{V}_0 + A_0 + B_0P_0$. And
\begin{align*}
    K &= P_0^\dagger \tilde{K}_0 +  P_0^\dagger B_0^\dagger A_0P_0 + \frac{1}{2} P_0^\dagger B_0^\dagger B_0P_0 + iP_0^\dagger H_0P_0\\
    &\hphantom{=}~+ \sum_{i \in I} P_i^\dagger (\idmat_{\mathtt{A}_i} \otimes B_i^\dagger) V_{ii}^{sc} P_i + \frac{1}{2}\sum_{i \in I} P_i^\dagger (\idmat_{\mathtt{A}_i} \otimes B_i^\dagger B_i) P_i + K_\mathcal{A} + iH_{\mathcal{A}^\prime},
\end{align*}
which has the desired form after defining $K_0 = \tilde{K}_0 + B_0^\dagger A_0P_0 + \frac{1}{2} B_0^\dagger B_0P_0 + iH_0P_0$.
\end{proof}

The representation in Part \ref{item:NormalFormAtomicLindblad} of Theorem \ref{THM:AtomicLindbladNormalForm} is not unique. The following theorem quantifies the freedom in that representation.  

\begin{theorem}{Theorem} \label{thm:UniquenessStatement}
The operators and spaces in Part \ref{item:NormalFormAtomicLindblad} of Theorem \ref{THM:AtomicLindbladNormalForm} can be chosen to satisfy the following minimality conditions: \labeltext{a)}{cond:minimalityLindblad}\, For all $i \in I$, the set $\left\{(X_{\mathtt{A}_i} \otimes \idmat_{\mathtt{F}_{ii}}) A_{ii} - A_{ii}X_{\mathtt{A}_i} \ket{\psi}  \,|\, X_{\mathtt{A_i}} \in \blt(\mathcal{H}_{\mathtt{A}_i}), \ket{\psi} \in \mathcal{H}_{\mathtt{A}_i} \right\}$ is total in $\mathcal{H}_{\mathtt{A}_i} \otimes \mathcal{H}_{\mathtt{F}_{ii}}$,
and \labeltext{b)}{cond:minimalityCP} the set $\left\{(X_{\mathtt{A}_i} \otimes \idmat_{\mathtt{F}_{ij}}) A_{ij} \ket{\psi} \,|\, X_{\mathtt{A_i}} \in \blt(\mathcal{H}_{\mathtt{A}_i}), \ket{\psi} \in \mathcal{H}_{\mathtt{A}_j} \right\}$ is total in $\mathcal{H}_{\mathtt{A}_i} \otimes \mathcal{H}_{\mathtt{F}_{ij}}$ for all $i, j \in I$ with $i \neq j$. \\
Let $L : \blt(\mathcal{H}) \rightarrow \blt(\mathcal{H})$ and $\tilde{L} : \blt(\mathcal{H}) \rightarrow \blt(\mathcal{H})$ be given by $L(X) = V^\dagger (X \otimes \idmat_\mathtt{E}) V - K^\dagger X - X K$ and $\tilde{L}(X) = \tilde{V}^\dagger (X \otimes \idmat_{\tilde{\mathtt{E}}}) \tilde{V} - \tilde{K}^\dagger X - X \tilde{K}$, with $V\in\blt (\mathcal{H};\mathcal{H}\otimes\mathcal{H}_{\mathtt{E}})$, $\tilde{V}\in\blt (\mathcal{H};\mathcal{H}\otimes\mathcal{H}_{\tilde{\mathtt{E}}})$, and $K,\tilde{K}\in\blt (\mathcal{H})$, and let $\mathcal{A}$ be an atomic *-subalgebra of $\blt(\mathcal{H})$, with decomposition given by Definition \ref{Defn:NormalFormAtomicAlgebra}. Suppose that $L(\mathcal{A}) \subseteq \mathcal{A}$ and $\tilde{L}(\mathcal{A}) \subseteq \mathcal{A}$ and let the corresponding representations (Theorem \ref{THM:AtomicLindbladNormalForm}) $(V_0, K_0, \{\mathcal{H}_{\mathtt{F}_{ij}} \},\{U_{ij}\}, \{A_{ij}\}, \{B_i\},\\ \{K_{\mathtt{A}_i}\}, \{H_{\mathtt{B}_i}\} )$ and $(\tilde{V}_0, \tilde{K}_0, \{\mathcal{H}_{\tilde{\mathtt{F}}_{ij}} \}, \{\tilde{U}_{ij}\}, \{\tilde{A}_{ij}\}, \{\tilde{B}_i\}, \{\tilde{K}_{\mathtt{A}_i}\}, \{\tilde{H}_{\mathtt{B}_i}\} )$ both satisfy conditions \ref{cond:minimalityLindblad} and \ref{cond:minimalityCP} above. Then, the following hold:
\begin{enumerate}
    \item \label{SubstitutionEquivalenceAlgebra}If $L(X_\mathcal{A}) = \tilde{L}(X_\mathcal{A})$ for all $X_\mathcal{A} \in \mathcal{A}$, then, for every $i \in I$, there exist a unitary $W_{ii} \in \blt(\mathcal{H}_{\mathtt{F}_{ii}}; \mathcal{H}_{\tilde{\mathtt{F}}_{ii}})$, vectors $\ket{\tilde{\psi}_i} \in \mathcal{H}_{\tilde{\mathtt{F}}_{ii}}$ with $\sup_{i \in I} \Vert \tilde{\psi}_i \Vert < \infty$, and numbers $\mu_i \in \mathbb{R}$, with $\sup_{i \in I} |\mu_i| < \infty$, such that $\tilde{A}_{ii} = (\idmat_{\mathtt{A}_i} \otimes W_{ii})A_{ii} + \idmat_{\mathtt{A}_i} \otimes\ket{\tilde{\psi}_i}$ and $\tilde{K}_{\mathtt{A}_i} = K_{\mathtt{A}_i} + (\idmat_{\mathtt{A}_i} \otimes \bra{\tilde{\psi_i}}W_{ii})A_{ii} + \frac{1}{2} \Vert \tilde{\psi}_i \Vert^2 + i\mu_i$. Moreover, for all $i, j \in I$ with $i \neq j$ there exists a unitary $W_{ij} \in \blt(\mathcal{H}_{\mathtt{F}_{ij}}; \mathcal{H}_{\tilde{\mathtt{F}}_{ij}})$ such that $\tilde{A}_{ij} = (\idmat_{\mathtt{A}_i} \otimes W_{ij})A_{ij}$.
    \item \label{SubstitutionEquivalenceAll}If $\tilde{V} = V$ and $\tilde{K} = K$, then $\tilde{V}_0 = V_0$, $\tilde{K}_0 = K_0$, $\tilde{B}_i = B_i - U_{ii} (\ket{W_{ii}^\dagger\tilde{\psi}_i} \otimes \idmat_{\mathtt{B}_i})$, and $\tilde{H}_{\mathtt{B}_i} = H_{\mathtt{B}_i} + \frac{i}{2} (G - G^\dagger) - \mu_i \idmat_{\mathtt{B}_i}$, where $G = B_i^\dagger U_{ii} (\ket{W_{ii}^\dagger \tilde{\psi}_i} \otimes \idmat_{\mathtt{B}_i})$. Moreover, for all $i, j \in I$,  we have  $\tilde{U}_{ij} = U_{ij}(W_{ij}^\dagger \otimes \idmat_{\mathtt{B}_i})$.
\end{enumerate}
\end{theorem}

\noindent \textbf{Remark.} It is the matter of a straightforward calculation to show that the (simultaneous) substitutions in \ref{SubstitutionEquivalenceAlgebra} and \ref{SubstitutionEquivalenceAll} above leave the operators $V$ and $K$ invariant. 
Thus Theorem \ref{thm:UniquenessStatement} quantifies exactly the freedom in our representation. Moreover, Theorem \ref{thm:GeneralUniquenessTheorem} quantifies the freedom in the choice of $(\mathcal{H}_\mathtt{E}, V, K)$. 

\begin{proof}
We start with proving the possibility of a reduction to minimality, as claimed in the first part of the theorem. 
Suppose $L$ is given according to Part \ref{item:NormalFormAtomicLindblad} of Theorem \ref{THM:AtomicLindbladNormalForm}, with data $(\tilde{V}_0, \tilde{K}_0, \{\mathcal{H}_{\tilde{\mathtt{F}}_{ij}} \}, \{\tilde{U}_{ij} \}, \{\tilde{A}_{ij}\}, \{\tilde{B}_i\}, \{\tilde{K}_{\mathtt{A}_i}\}, \{\tilde{H}_{\mathtt{B}_i}\} )$. For any $i \in I$, $X_{\mathtt{A}_i} \in \blt(\mathcal{H}_{\mathtt{A}_i})$, we have 
$P_i L(P_i^\dagger (X_{\mathtt{A}_i} \otimes \idmat_{\mathtt{B}_i})P_i) P_i^\dagger = \left[ \tilde{A}_{ii}^\dagger (X_{\mathtt{A}_i} \otimes \idmat_{\tilde{\mathtt{F}}_{ii}}) \tilde{A}_{ii} - \tilde{K}_{\mathtt{A}_i}^\dagger X_{\mathtt{A}_i} - X_{\mathtt{A}_i} \tilde{K}_{\mathtt{A}_i} \right] \otimes \idmat_{\mathtt{B}_i} =: L_{ii}^\downarrow(X_{\mathtt{A}_i}) \otimes \idmat_{\mathtt{B}_i}$.
By Theorem \ref{thm:GeneralUniquenessTheorem}, there exists $(\mathcal{H}_{\mathtt{F}_{ii}}, A_{ii}, K_{\mathtt{A}_i})$ such that $L_{ii}^\downarrow(X_{\mathtt{A}_i}) = A_{ii}^\dagger (X_{\mathtt{A}_i} \otimes \idmat_{\mathtt{F}_{ii}}) A_{ii} - K_{\mathtt{A}_i}^\dagger X_{\mathtt{A}_i} - X_{\mathtt{A}_i} K_{\mathtt{A}_i}$, $A_{ii}$ satisfies condition \ref{cond:minimalityLindblad}, $\tilde{A}_{ii} = (\idmat_{\mathtt{A}_i} \otimes W_{ii})A_{ii} + \idmat_{\mathtt{A}_i} \otimes \ket{\tilde{\psi}_i}$, and $\tilde{K}_{\mathtt{A}_i} = K_{\mathtt{A}_i} + (\idmat_{\mathtt{A}_i} \otimes \bra{\tilde{\psi}_i} W_{ii})A_{ii} + \frac{1}{2}\Vert \tilde{\psi}_i \Vert^2 + i\mu_i$, for an isometry $W_{ii} \in \blt(\mathcal{H}_{\mathtt{F}_{ii}}; \mathcal{H}_{\tilde{\mathtt{F}}_{ii}})$, a vector $\ket{\tilde{\psi}_i} \in \mathcal{H}_{\tilde{\mathtt{F}}_{ii}}$, and a number $\mu_i \in \mathbb{R}$. 
We define $U_{ii} = \tilde{U}_{ii}(W_{ii} \otimes \idmat_{\mathtt{B}_i})$, so that $U_{ii}$ is an isometry.
Furthermore, we define $B_i = \tilde{B_i} + \ket{\tilde{\psi}_i} \otimes \idmat_{\mathtt{B}_i}$ and $H_{\mathtt{B}_i} = \tilde{H}_{\mathtt{B}_i} + \frac{i}{2} (\tilde{G} - \tilde{G}^\dagger) + \mu_i \idmat_{\mathtt{B}_i}$, where $\tilde{G} = \tilde{B}_i^\dagger \tilde{U}_{ii}(\ket{\tilde{\psi}_i} \otimes \idmat_{\mathtt{B}_i})$.
A direct calculation shows that replacing the operators with `tilde' by the ones without, does not change $V$ and $K$, but now $A_{ii}$ satisfies condition \ref{cond:minimalityLindblad}. The claim that $\sup_{i \in I} \Vert \tilde{\psi}_i \Vert < \infty$ and $\sup_{i \in I} |\mu_i| < \infty$ follows, since $L$ would be unbounded otherwise. For $i, j \in I$ with $i \neq j$, we have $P_jL(P_i^\dagger(X_{\mathtt{A}_i} \otimes \idmat_{\mathtt{B}_i})P_i)P_j^\dagger = \left[ \tilde{A}_{ij}^\dagger (X_{\mathtt{A}_i} \otimes \idmat_{\tilde{\mathtt{F}}_{ij}}) \tilde{A}_{ij} \right] \otimes \idmat_{\mathtt{B}_j} =: L_{ij}^\downarrow(X_{\mathtt{A}_i}) \otimes \idmat_{\mathtt{B}_j}$. 
By Theorem \ref{thm:UniquenessStinespringThm}, there exists $(\mathcal{H}_{\tilde{\mathtt{F}}_{ij}}, A_{ij})$ such that $L_{ij}^\downarrow(X_{\mathtt{A}_i}) = A_{ij}^\dagger (X_{\mathtt{A}_i} \otimes \idmat_{\mathtt{F}_{ij}})A_{ij}$, condition \ref{cond:minimalityCP} holds, and $\tilde{A}_{ij} = (\idmat_{\mathtt{A}_i} \otimes W_{ij})A_{ij}$, for some isometry $W_{ij} \in \blt(\mathcal{H}_{\mathtt{F}_{ij}}; \mathcal{H}_{\tilde{\mathtt{F}}_{ij}})$. 
We define $U_{ij} = \tilde{U}_{ij}(W_{ij} \otimes \idmat_{\mathtt{B}_j})$. 
It follows that $U_{ij}$ is an isometry and also that $U_{ik}^\dagger U_{il} = \delta_{kl} \idmat$. Again, a calculation shows that replacing the operators with `tilde' by the ones without, does not change $K$ and $V$.

Next, we want to prove \ref{SubstitutionEquivalenceAlgebra}. Since $L(X_\mathcal{A}) = \tilde{L}(X_\mathcal{A})$ for all $X_\mathcal{A} \in \mathcal{A}$, we have in particular $P_i L(P_i^\dagger (X_{\mathtt{A}_i} \otimes \idmat_{\mathtt{B}_i})P_i) P_i^\dagger = P_i \tilde{L}(P_i^\dagger (X_{\mathtt{A}_i} \otimes \idmat_{\mathtt{B}_i})P_i) P_i^\dagger$ for all $i \in I$. 
This is equivalent to $A_{ii}^\dagger (X_{\mathtt{A}_i} \otimes \idmat_{\mathtt{F}_{ii}}) A_{ii} - K_{\mathtt{A}_i}^\dagger X_{\mathtt{A}_i} - X_{\mathtt{A}_i} K_{\mathtt{A}_i} = \tilde{A}_{ii}^\dagger (X_{\mathtt{A}_i} \otimes \idmat_{\tilde{\mathtt{F}}_{ii}}) \tilde{A}_{ii} - \tilde{K}_{\mathtt{A}_i}^\dagger X_{\mathtt{A}_i} - X_{\mathtt{A}_i} \tilde{K}_{\mathtt{A}_i}$.
Since \ref{cond:minimalityLindblad} holds for $L$ and $\tilde{L}$, Theorem \ref{thm:GeneralUniquenessTheorem} implies the existence of unitaries $W_{ii} \in \blt(\mathcal{H}_{\mathtt{F}_{ii}}; \mathcal{H}_{\tilde{\mathtt{F}}_{ii}})$, vectors $\ket{\tilde{\psi}_i} \in \mathcal{H}_{\tilde{\mathtt{F}}_{ii}}$, and numbers $\mu_i \in \mathbb{R}$ s.t.~$\tilde{A}_{ii} = (\idmat_{\mathtt{A}_i} \otimes W_{ii})A_{ii} + \idmat_{\mathtt{A}_i} \otimes \ket{\tilde{\psi}_i}$ and $\tilde{K}_{\mathtt{A}_i} = K_{\mathtt{A}_i} + (\idmat_{\mathtt{A}_i} \otimes \bra{\tilde{\psi_i}}W_{ii})A_{ii} + \frac{1}{2} \Vert \tilde{\psi}_i \Vert^2 + i\mu_i$. The claim that $\sup_{i \in I} \Vert \tilde{\psi}_i \Vert < \infty$ and $\sup_{i \in I} |\mu_i| < \infty$ follows, since $\tilde{L}$ would be unbounded otherwise. For $i, j \in I$ with $i \neq j$ we have $P_j L(P_i^\dagger (X_{\mathtt{A}_i} \otimes \idmat_{\mathtt{B}_i})P_i) P_j^\dagger = P_j \tilde{L}(P_i^\dagger (X_{\mathtt{A}_i} \otimes \idmat_{\mathtt{B}_i})P_i) P_j^\dagger$, which is equivalent to $A_{ij}^\dagger(X_{\mathtt{A}_i} \otimes \idmat_{\mathtt{F}_{ij}})A_{ij} = \tilde{A}_{ij}^\dagger(X_{\mathtt{A}_i} \otimes \idmat_{\tilde{\mathtt{F}}_{ij}})\tilde{A}_{ij}$ for all $X_{\mathtt{A}_i} \in \blt(\mathcal{H}_{\mathtt{A}_i})$. Since \ref{cond:minimalityCP} holds for $L$ and $\tilde{L}$, Theorem \ref{thm:UniquenessStinespringThm} implies the existence of a unitary $W_{ij} \in \blt(\mathcal{H}_{\mathtt{F}_{ij}}; \mathcal{H}_{\tilde{\mathtt{F}}_{ij}})$ s.t.~$\tilde{A}_{ij} = (\idmat_{\mathtt{A}_i} \otimes W_{ij})A_{ij}$. This is claim \ref{SubstitutionEquivalenceAlgebra}.

For part \ref{SubstitutionEquivalenceAll}, we first notice that $\tilde{V} = V$ and $\tilde{K} = K$ immediately implies (by projecting into the respective subspace) that $\tilde{V}_0 = V_0$ and $\tilde{K}_0 = K_0$. Moreover, for any $i\in I$,
\begin{equation} \label{eq:Projection-WiseEquivalenceLindblad}
    V_{ii} := (\idmat_{\mathtt{A}_{ii}} \otimes U_{ii})(A_{ii} \otimes \idmat_{\mathtt{B}_{ii}}) + \idmat_{\mathtt{A}_{i}} \otimes B_i = (\idmat_{\mathtt{A}_{ii}} \otimes \tilde{U}_{ii})(\tilde{A}_{ii} \otimes \idmat_{\mathtt{B}_{ii}}) + \idmat_{\mathtt{A}_{i}} \otimes \tilde{B}_i =: \tilde{V}_{ii}\, .
\end{equation}
Thus, 
\begin{align}\label{Eq:expansionSomething}
    &(X_{\mathtt{A}_i} \otimes \idmat_{\mathtt{B}_i\mathtt{E}})\tilde{V}_{ii} - \tilde{V}_{ii}(X_{\mathtt{A}_i} \otimes \idmat_{\mathtt{B}_i}) = (\idmat_{\mathtt{A}_i} \otimes \tilde{U}_{ii})([ (X_{\mathtt{A}_i} \otimes \idmat_{\tilde{\mathtt{F}}_{ii}})\tilde{A}_{ii} - \tilde{A}_{ii}X_{\mathtt{A}_i}  ] \otimes \idmat_{\mathtt{B}_i}) \nonumber \\ 
    &= (\idmat_{\mathtt{A}_i} \otimes (\tilde{U}_{ii}(W_{ii} \otimes \idmat_{\mathtt{B}_i})))([ (X_{\mathtt{A}_i} \otimes \idmat_{\mathtt{F}_{ii}})A_{ii} - A_{ii}X_{\mathtt{A}_i}  ] \otimes \idmat_{\mathtt{B}_i}) \nonumber \\ 
    &= (\idmat_{\mathtt{A}_i} \otimes U_{ii})([ (X_{\mathtt{A}_i} \otimes \idmat_{\mathtt{F}_{ii}})A_{ii} - A_{ii}X_{\mathtt{A}_i}  ] \otimes \idmat_{\mathtt{B}_i})\, ,
\end{align} 
where the second line was obtained by using the relation between $A_{ii}$ and $\tilde{A}_{ii}$ in Part \ref{SubstitutionEquivalenceAlgebra}. 
From the equality of the last two lines and the totality implied by \ref{cond:minimalityLindblad}, we conclude $U_{ii} = \tilde{U}_{ii}(W_{ii} \otimes \idmat_{\mathtt{B}_i})$.
Using this relation, the relation between $A_{ii}$ and $\tilde{A}_{ii}$, and Eq.~\eqref{eq:Projection-WiseEquivalenceLindblad} yields $\tilde{B}_i = B_i - U_{ii} (\ket{W_{ii}^\dagger \tilde{\psi}_i} \otimes \idmat_{\mathtt{B}_i})$.
Moreover, from $P_i K P_i^\dagger = (\idmat_{\mathtt{A}_i} \otimes B_i^\dagger)(\idmat_{\mathtt{A}_i} \otimes U_{ii})(A_{ii} \otimes \idmat_{\mathtt{B}_i}) + \frac{1}{2} (\idmat_{\mathtt{A}_i} \otimes B_i^\dagger B_i) + (K_{\mathtt{A}_i} \otimes \idmat_{\mathtt{B}_i}) + (\idmat_{\mathtt{A}_i} \otimes iH_{\mathtt{B}_i}) = (\idmat_{\mathtt{A}_i} \otimes \tilde{B}_i^\dagger)(\idmat_{\mathtt{A}_i} \otimes \tilde{U}_{ii})(\tilde{A}_{ii} \otimes \idmat_{\mathtt{B}_i}) + \frac{1}{2} (\idmat_{\mathtt{A}_i} \otimes \tilde{B}_i^\dagger \tilde{B}_i) + (\tilde{K}_{\mathtt{A}_i} \otimes \idmat_{\mathtt{B}_i}) + (\idmat_{\mathtt{A}_i} \otimes i\tilde{H}_{\mathtt{B}_i}) = P_i \tilde{K} P_i^\dagger$ and the already established relations between the operators with and without 'tilde', we obtain $\tilde{H}_{\mathtt{B}_i} = H_{\mathtt{B}_i} + \frac{i}{2} (G - G^\dagger) - \mu_i \idmat_{\mathtt{B}_i}$.
Finally, for $i, j \in I$ with $i \neq j$ we get $(X_{\mathtt{A}_i} \otimes \idmat_{\mathtt{B}_i\mathtt{E}})V_{ij} = (\idmat_{\mathtt{A}_i} \otimes U_{ij})([(X_{\mathtt{A}_i} \otimes \idmat_{\mathtt{F}_{ij}})A_{ij}] \otimes \idmat_{\mathtt{B}_j}) = (\idmat_{\mathtt{A}_i} \otimes \tilde{U}_{ij}) ([(X_{\mathtt{A}_i} \otimes \idmat_{\mathtt{F}_{ij}}) \tilde{A}_{ij}] \otimes \idmat_{\mathtt{B}_j}) = (\idmat_{\mathtt{A}_i} \otimes (\tilde{U}_{ij}(W_{ij} \otimes \idmat_{\mathtt{B}_j})))([(X_{\mathtt{A}_i} \otimes \idmat_{\mathtt{F}_{ij}})A_{ij}] \otimes \idmat_{\mathtt{B}_j})$.
By the totality condition \ref{cond:minimalityCP}, we can conclude $U_{ij} = \tilde{U}_{ij}(W_{ij} \otimes \idmat_{\mathtt{B}_j})$. Since $W_{ij}$ is unitary, this finishes the proof.
\end{proof}

For later convenience, we also note the following. 
\begin{theorem}{Corollary} \label{cor:K-onlyPartCorollary}
Let $L : \blt(\mathcal{H}) \rightarrow \blt(\mathcal{H})$ be given by $L(X) = - K^\dagger X - X K$, with $K\in\blt (\mathcal{H})$, and let $\mathcal{A}$ be an atomic *-subalgebra of $\blt(\mathcal{H})$, with decomposition given by Definition \ref{Defn:NormalFormAtomicAlgebra}. If $L(\mathcal{A}) \subseteq \mathcal{A}$, then we can choose $A_{ij} = 0$ and $B_i = 0$, for all $i, j \in I$ in the corresponding representation of Part \ref{item:NormalFormAtomicLindblad} 
of Theorem \ref{THM:AtomicLindbladNormalForm}. Thus, $K = K_\mathcal{A} + iH_{\mathcal{A}^\prime} + P_0^\dagger K_0$.
\end{theorem}
\begin{proof}
If $V = 0$ and $K$ are given via $(V_0, K_0, \{\mathcal{H}_{\mathtt{F}_{ij}} \}, \{U_{ij}\}, \{A_{ij}\}, \{B_i\}, \{K_{\mathtt{A}_i}\}, \{H_{\mathtt{B}_i}\} )$ such that \ref{cond:minimalityLindblad} and \ref{cond:minimalityCP} in Theorem \ref{thm:UniquenessStatement} hold, $0 = V_{ij} = (P_i \otimes \idmat_\mathtt{E})VP_j$ implies $(\idmat_{\mathtt{A}_i} \otimes U_{ij})([(X_{\mathtt{A}_i} \otimes \idmat_{\mathtt{F}_{ij}})A_{ij}] \otimes \idmat_{\mathtt{B}_j}) = 0$ for $i \neq j$ and $(\idmat_{\mathtt{A}_i} \otimes U_{ii})([ (X_{\mathtt{A}_i} \otimes \idmat_{\mathtt{F}_{ii}})A_{ii} - A_{ii}X_{\mathtt{A}_i}  ] \otimes \idmat_{\mathtt{B}_i})=0$ for all $i \in I$ (compare Eqs.~\eqref{eq:Projection-WiseEquivalenceLindblad} and \eqref{Eq:expansionSomething}). Thus by the totality conditions \ref{cond:minimalityLindblad} and \ref{cond:minimalityCP}, we conclude $U_{ij} = 0$ for all $i, j \in I$, which implies that $\mathcal{H}_{\mathtt{F}_{ij}}$ is zero-dimensional. Hence, $A_{ij} = 0$ and consequently also $B_i = 0$ for all $i \in I$. 
\end{proof}

\section{Applications} \label{applicationSection}

\subsection{Semicausal quantum dynamical semigroups}
As a first application of our results, we use them to reprove the main result of~\cite{hasenoehrl2021quantum}, namely the characterization of GKLS generators of semicausal quantum dynamical semigroups, a crucial step towards characterizing the generators of continuous one-parameter semigroups of quantum superchannels.
Here, we call a CP-map $\Phi:\blt (\mathcal{H}_{\mathtt{A}}\otimes\mathcal{H}_{\mathtt{B}})\to \blt (\mathcal{H}_{\mathtt{A}}\otimes\mathcal{H}_{\mathtt{B}})$ semicausal if there is a CP-map $\Phi_\mathtt{A}: \blt (\mathcal{H}_{\mathtt{A}}) \to \blt (\mathcal{H}_{\mathtt{A}})$ such that $\Phi (X_\mathtt{A} \otimes\idmat_{\mathtt{B}}) = \Phi_\mathtt{A} (X_\mathtt{A}) \otimes\idmat_{\mathtt{B}}$ holds for all $X_\mathtt{A}\in \blt (\mathcal{H}_{\mathtt{A}})$.
That is, $\Phi$ is semicausal if and only if $\Phi (\mathcal{A}_{sc})\subseteq\mathcal{A}_{sc}$ holds for the atomic vN-subalgebra $\mathcal{A}_{sc}:= \blt (\mathcal{H}_{\mathtt{A}})\otimes\idmat_\mathtt{B}\subseteq\blt (\mathcal{H}_{\mathtt{A}}\otimes\mathcal{H}_{\mathtt{B}})$.

Using Theorem~\ref{THM:AtomicLindbladNormalForm}, we see that a GKLS generator $L: \blt (\mathcal{H}_{\mathtt{A}}\otimes\mathcal{H}_{\mathtt{B}})\to \blt (\mathcal{H}_{\mathtt{A}}\otimes\mathcal{H}_{\mathtt{B}})$, $L(X)=V^\dagger (X\otimes\idmat_\mathtt{E})V - K^\dagger X - XK$, satisfies $L(\mathcal{A}_{sc})\subseteq \mathcal{A}_{sc}$ if and only if there exist a Hilbert space $\mathcal{H}_\mathtt{F}$, an operator $A\in\blt (\mathcal{H}_\mathtt{A};\mathcal{H}_\mathtt{A}\otimes \mathcal{H}_\mathtt{F})$, and an isometry $U\in\blt (\mathcal{H}_\mathtt{F}\otimes \mathcal{H}_\mathtt{B}; \mathcal{H}_\mathtt{B}\otimes \mathcal{H}_\mathtt{E})$; operators $B\in \blt(\mathcal{H}_\mathtt{B}; \mathcal{H}_\mathtt{B}\otimes \mathcal{H}_\mathtt{E})$, $K_\mathtt{A}\in \blt (\mathcal{H}_\mathtt{A})$, and a self-adjoint operator $H_\mathtt{B}\in \blt (\mathcal{H}_\mathtt{B})$, such that $V=(\idmat_{\mathtt{A}}\otimes U)(A\otimes\idmat_{\mathtt{B}}) + \idmat_{\mathtt{A}}\otimes B$ and $K=(\idmat_{\mathtt{A}}\otimes B^\dagger U)(A\otimes\idmat_{\mathtt{B}}) + \tfrac{1}{2}\idmat_{\mathtt{A}}\otimes B^\dagger B + K_\mathtt{A}\otimes\idmat_{\mathtt{B}} + \idmat_{\mathtt{A}}\otimes i H_\mathtt{B}$.
This is exactly~\cite[Theorem V.6]{hasenoehrl2021quantum}.

\subsection{Quantum dynamical semigroups with an atomic decoherence-free subalgebra}

For our second application, we first recall that for a uniformly continuous and unital quantum dynamical semigroup $\mathcal{T}=(T_t)_{t\geq 0}$ acting on $\blt (\mathcal{H})$ the decoherence-free subalgebra $\mathcal{N}(\mathcal{T})$ is the largest vN-subalgebra of $\blt (\mathcal{H})$ on which every $T_t$ acts as a $\ast$-automorphism~\cite[Theorem $3$]{doi:10.1063/1.5030954} or, equivalently, such that $T_t(X_{\mathcal{N}(\mathcal{T})}) = e^{i\tilde{H}t} X_{\mathcal{N}(\mathcal{T})} e^{-i\tilde{H}t}$ for all $X_{\mathcal{N}(\mathcal{T})} \in \mathcal{N}(\mathcal{T})$, where $\tilde{H} \in \blt(\mathcal{H})$ is self-adjoint~\cite[Proposition $2$]{doi:10.1063/1.5030954}.
In particular, note that every $T_t$ leaves $\mathcal{N}(\mathcal{T})$ invariant.
As shown in~\cite{doi:10.1142/S0129055X16500033}, a quantum dynamical semigroup inherits structure from its decoherence-free subalgebra. This is the content of the following: 

\begin{theorem}{Theorem}{\cite[Theorem $3.2(1)$]{doi:10.1142/S0129055X16500033}}\label{theorem:gkls-decoherence-free-subalgebra}
   Let $\mathcal{T}$ and $\mathcal{N}(\mathcal{T})$ be as above, and assume that $\mathcal{N}(\mathcal{T})$ is atomic, with normal form
   as in Definition~\ref{Defn:NormalFormAtomicAlgebra}.
   Then, for any GKLS-generator $L$, given in Kraus form by $L(X) = \sum_{n \in N} \phi_n^\dagger X \phi_n - K^\dagger X - XK$, where $K = \tfrac{1}{2} \sum_{n \in N} \phi_n^\dagger \phi_n + i\operatorname{Im}(K)$, we have
   \begin{equation}\label{eq:decoherence-free-gkls}
       \phi_n 
       = \sum_{i \in I} P_i^\dagger (\idmat_{\mathtt{A}_i} \otimes \beta_{n;i})P_i,\quad
       \operatorname{Im}(K)
       = \sum_{i \in I} P_i^\dagger (\kappa_{\mathtt{A}_i}\otimes\idmat_{\mathtt{B}_i} + \idmat_{\mathtt{A}_i}\otimes \kappa_{\mathtt{B}_i}) P_i,
   \end{equation}
   for some $\beta_{n;i}\in\blt (\mathcal{H}_{\mathtt{B}_i})$ and some self-adjoint $\kappa_{\mathtt{A}_i}\in\blt(\mathcal{H}_{\mathtt{A}_i})$ and $\kappa_{\mathtt{B}_i}\in \blt(\mathcal{H}_{\mathtt{B}_i})$.
\end{theorem}

We can recover Theorem~\ref{theorem:gkls-decoherence-free-subalgebra} as a special case of our results as follows. 
By the characterization of the decoherence-free subalgebra described above, we have $L(X_{\mathcal{N}(\mathcal{T})}) = i[\tilde{H},X_{\mathcal{N}(\mathcal{T})}]\in \mathcal{N}(\mathcal{T})$ for all $X_{\mathcal{N}(\mathcal{T})}\in \mathcal{N}(\mathcal{T})$.
Therefore, if we define $\tilde{L}(X):= -(i\tilde{H})^\dagger X - X (i\tilde{H})$ for all $X\in \blt (\mathcal{H})$, then Corollary~\ref{cor:K-onlyPartCorollary} implies $\tilde{H} = \sum_{i \in I} P_i^\dagger (\tilde{H}_{\mathtt{A}_i} \otimes \idmat_{\mathtt{B}_i} + \idmat_{\mathtt{A}_i} \otimes \tilde{H}_{\mathtt{B}_i}) P_i$ for self-adjoint $\tilde{H}_{\mathtt{A}_i}\in\blt (\mathcal{H}_{\mathtt{A}_i})$ and $\tilde{H}_{\mathtt{B}_i}\in\blt (\mathcal{H}_{\mathtt{B}_i})$.
This already provides a normal form representation of $\tilde{L}$ as in Theorem~\ref{THM:AtomicLindbladNormalForm}, and the minimality conditions of Theorem~\ref{thm:UniquenessStatement} are satisfied, since $\tilde{A}_{ij} = 0$. 
Now, if $L$ is w.l.o.g.~also minimal, Part \ref{SubstitutionEquivalenceAlgebra} of Theorem~\ref{thm:UniquenessStatement} (with the roles of the operators with and without `tilde' interchanged) implies that $A_{ii} = \mathds{1}_{\mathtt{A}_i} \otimes \ket{\psi_i}$ for some vectors $\ket{\psi_i}\in \mathcal{H}_{\mathtt{F}_{ii}}$ and that $A_{ij}=0$ for $i\neq j$.
On the level of the Kraus operators, this yields $\phi_n = (\mathds{1}\otimes \bra{e_n})\big(\sum_{i\in I} (P_i^\dagger \otimes\mathds{1}_{\mathtt{E}}) (\mathds{1}_{\mathtt{A}_{i}} \otimes U_{ii} (\ket{\psi_i}\otimes \mathds{1}_{\mathtt{B}_i}) + \mathds{1}_{\mathtt{A}_{i}}\otimes B_i ) P_i \big)$, which has the desired form with $\beta_{n;i} = (\mathds{1}_{\mathtt{B}_i}\otimes \bra{e_n})( U_{ii} (\ket{\psi_i}\otimes \mathds{1}_{\mathtt{B}_i}) + B_i)$.
(Note: If $L$ is not already minimal, we can first follow the steps in the proof of Theorem \ref{thm:UniquenessStatement} to reduce to a minimal generator and then apply the above reasoning.)
Moreover, using the representation of Theorem~\ref{THM:AtomicLindbladNormalForm} and
again Part \ref{SubstitutionEquivalenceAlgebra} of Theorem~\ref{thm:UniquenessStatement}, 
\begin{align*}
    \operatorname{Im}(K)
    &= \sum_{i\in I}P_i^\dagger \big( \operatorname{Im}( (\mathds{1}_{\mathtt{A}_i}\otimes B_i^\dagger U_{ii}) (A_{ii}\otimes \mathds{1}_{\mathtt{B}_i})) + \operatorname{Im}(K_{\mathtt{A}_i})\otimes \mathds{1}_{\mathtt{B}_i} + \mathds{1}_{\mathtt{A}_i}\otimes H_{\mathtt{B}_i} \big) P_i\\
    &= \sum_{i\in I}P_i^\dagger \big( \mathds{1}_{\mathtt{A}_i}\otimes \operatorname{Im}( B_i^\dagger U_{ii} (\ket{\psi_i}\otimes \mathds{1}_{\mathtt{B}_i})) + \operatorname{Im}(\tilde{K}_{\mathtt{A}_i})\otimes \mathds{1}_{\mathtt{B}_i} + \mu_i + \mathds{1}_{\mathtt{A}_i}\otimes H_{\mathtt{B}_i} \big) P_i ,
\end{align*}
which has the desired form with $\kappa_{\mathtt{A}_i} := \operatorname{Im}(\tilde{K}_{\mathtt{A}_i})$ and $\kappa_{\mathtt{B}_i}:= H_{\mathtt{B}_i} + \operatorname{Im}(B_i^\dagger U_{ii} (\ket{\psi_i}\otimes \mathds{1}_{\mathtt{B}_i})) + \mu_i$.

As a final remark on our short discussion of the decoherence-free subalgebra, we point out that~\cite[Corollary $21$]{sasso2021general} showed that $\mathcal{N}(\mathcal{T})$ is atomic whenever the quantum dynamical semigroup $\mathcal{T}$ admits a normal faithful invariant state. 
In many situations of interest, we can therefore focus on $\mathcal{N}(\mathcal{T})$ being atomic, as in Theorems~\ref{THM:AtomicLindbladNormalForm} and~\ref{theorem:gkls-decoherence-free-subalgebra}.

\subsection{Quantum dynamical semigroups and CP-maps with an invariant maximal abelian subalgebra}

Our third application is concerned with  the following question: Given a maximal abelian vN-subalgebra $\mathcal{C}$ of $\blt (\mathcal{H})$, that is $\mathcal{C}^\prime = \mathcal{C}$, what is the most general form of a GKLS-generator 
that leaves $\mathcal{C}$ invariant? 
According to Theorem~\ref{theorem:gkls-invariant-afd-algebra}, we can reduce the above question to characterizing CP-maps with an invariant maximal abelian vN-subalgebra.
The latter question was previously investigated in~\cite{doi:10.1007/s11080-005-0485-3, FrancoFagnola2007}. More precisely, \cite[Theorem $1$]{doi:10.1007/s11080-005-0485-3} gave an abstract characterization of such GKLS-generators in terms of a commutation relation and a sufficient condition on the Kraus operators of the CP part of the GKLS-generator.
Ref.~\cite{FrancoFagnola2007} extended these deliberations and, in particular, gave a necessary and sufficient condition for a normal CP-map to leave a maximal abelian vN-subalgebra invariant:

\begin{theorem}{Theorem}{\cite[Corollary $3.4$]{FrancoFagnola2007}}\label{corollary:cp-map-invariant-maximal-abelian-subalgebra}
    Let $\Phi$ be a normal CP-map on $\blt (\mathcal{H})$ with Kraus decomposition $\Phi(X)=\sum_{n\in N} \phi_n^\dagger X \phi_n$. 
    Let $\mathcal{C}$ be a maximal abelian vN-subal\-ge\-bra of $\blt (\mathcal{H})$.
    Then, $\Phi$ leaves $\mathcal{C}$ invariant if and only if for every $c\in\mathcal{C}$ there exist $c_{mn}=c_{mn}(c)\in\mathcal{C}$ for $m,n\in N$ s.t.~\labeltext{1)}{cond:maximally-abelian-condition-1} $c_{mn}(c^\dagger) = c_{nm}(c)^\dagger$ and \labeltext{2)}{cond:maximally-abelian-condition-2} $[c,\phi_m] = \sum_{n\in N} c_{mn} \phi_n$.
\end{theorem}

The ``if''-direction in Theorem~\ref{corollary:cp-map-invariant-maximal-abelian-subalgebra} can be seen using the von Neumann bicommutant theorem. 
We can recover the ``only if''-direction, albeit only for atomic $\mathcal{C}$, as a consequence of Theorem~\ref{theorem:cp-map-algebra-to-algebra-stinespring} as follows:
If $\mathcal{C}$ is a maximal abelian and atomic vN-subal\-ge\-bra of $\blt (\mathcal{H})$, then its decomposition as in Definition~\ref{Defn:NormalFormAtomicAlgebra} becomes particularly simple, with $\operatorname{dim}(\mathcal{H}_{\mathtt{A}_i})=1=\operatorname{dim}(\mathcal{H}_{\mathtt{B}_i})$ for all $i\in I$. Thus, for any $i\in I$, there exists $\ket{p_i}\in\mathcal{H}$ such that $P_i = \bra{p_i}\in\blt (\mathcal{H};\mathbb{C})$ such that $\{\ket{p_i}\}_{i \in I}$ forms an orthonormal basis of $\mathcal{H}$.
The decomposition of $V$ in Theorem~\ref{theorem:cp-map-algebra-to-algebra-stinespring} in this case simplifies to $V = \sum_{i,j\in I} \ket{p_i}\bra{p_j} \otimes\ket{\psi_{ij}}$ for some vectors $\ket{\psi_{ij}}\in \mathcal{H}_{\mathtt{E}}$ satisfying $\braket{\psi_{ij}}{\psi_{ik}}=0$ for all $i,j,k\in I$ with $j\neq k$. 
Accordingly, the Kraus operators of $\Phi$ can be written as $\phi_n = \sum_{i,j\in I} \braket{e_n}{\psi_{ij}}\ket{p_i}\bra{p_j}$ for some orthonormal basis $\{\ket{e_n}\}_{n\in N}$ of $\mathcal{H}_{\mathtt{E}}$. 

Since any $c \in \mathcal{C}$ can be decomposed as $c = \sum_{i\in I} c_i\ket{p_i}\bra{p_i}$, we obtain for any $m\in N$:
\begin{equation}\label{eq:commutator-maximally-abelian-kraus}
    [c, \phi_m]
    = \sum_{i,j\in I} (c_i - c_j)\braket{e_m}{\psi_{ij}}\ket{p_i}\bra{p_j}\, .
\end{equation}
As $\braket{\psi_{ij}}{\psi_{ik}}=0$ for $j\neq k$, we can define $C_i\in \blt (\mathcal{H}_\mathtt{E})$ by linearly extending
\begin{equation}\label{eq:defining-relation-linear-maps-maximally-abelian}
    C_i \ket{\psi_{ij}}
    = (c_i - c_j)\ket{\psi_{ij}},\quad\text{ for all } i,j\in I\, .
\end{equation}
We consider a basis expansion $C_i:= \sum_{m,n\in N} c_{mn;i}\ket{e_m}\bra{e_n}$ and define $c_{mn}=\sum_{i\in I} c_{mn;i}\ket{p_i}\bra{p_i} \in\mathcal{C}$.
Then, we get
\begin{align*}
    \sum_{n\in N} c_{mn} \phi_n
    &= \sum_{n\in N}\sum_{i,j\in I} \braket{e_n}{\psi_{ij}} c_{mn;i}\ket{p_i}\bra{p_j}\\
    &= \sum_{i,j\in I}\ket{p_i}\bra{p_j}\cdot \underbrace{\sum_{n\in N}\braket{e_n}{\psi_{ij}} c_{mn;i}}_{=\bra{e_m}C_i\ket{\psi_{ij}} \overset{\eqref{eq:defining-relation-linear-maps-maximally-abelian}}{=} (c_i - c_j)\braket{e_m}{\psi_{ij}}}\\
    &= \sum_{i,j\in I} (c_i - c_j)\braket{e_m}{\psi_{ij}}\ket{p_i}\bra{p_j}\, ,
\end{align*}
which equals $[c, \phi_m]$ by Eq.~\eqref{eq:commutator-maximally-abelian-kraus}.
Thus,~\ref{cond:maximally-abelian-condition-2} is satisfied. Also~\ref{cond:maximally-abelian-condition-1} holds, since the replacement $c\mapsto c^\dagger$ in the above reasoning leads to $c_i\mapsto \Bar{c}_i$, which in turn gives $C_i\mapsto C_i^\dagger$, finally implying $c_{mn;i}\mapsto \Bar{c}_{nm;i}$ and thus $c_{mn}\mapsto c_{nm}^\dagger$.

We note that our Theorem~\ref{theorem:cp-map-algebra-to-algebra-stinespring} not only gives rise to Theorem~\ref{corollary:cp-map-invariant-maximal-abelian-subalgebra}, but also provides a concrete characterization of the most general CP-maps satisfying the criterion of Theorem~\ref{corollary:cp-map-invariant-maximal-abelian-subalgebra} for the atomic case. 
Similarly, we can reproduce and concretize~\cite[Theorem $1.2$]{rajarama2008maximal} for atomic $\mathcal{C}$.

\subsection{Completely positive and trace-preserving maps with fixed points}

In this section, we look at the Koashi-Imoto Theorem \cite[Eq.~(85) and Theorem 3]{PhysRevA.66.022318} and restrict ourselves to finite-dimensional systems. 
The Koashi-Imoto Theorem characterizes the form of CP-maps $T \in \CP(\mathcal{H})$ that are trace-preserving and state-preserving, in the sense that $T(\rho_s) = \rho_s$ for some set of density matrices $\{\rho_s\}_{s \in S}$. 
If $T(\rho) = \ptr{\mathtt{E}}{V\rho V^\dagger}$, with $V \in \blt(\mathcal{H}; \mathcal{H} \otimes \mathcal{H}_{\mathtt{E}})$; $\mathcal{F}_T$ is the set of fixed points of $T$; $\mathcal{H}_{\mathcal{F}_T} := \cup_{\rho \in \mathcal{F}_T} \mathrm{supp}(\rho) \subseteq \mathcal{H}$ is the support of a maximal rank fixed point; and $Q : \mathcal{H} \rightarrow \mathcal{H}_{\mathcal{F}_T}$ is the corresponding projection, then the Koashi-Imoto Theorem states that there is a decomposition of $\mathcal{H}_{\mathcal{F}_T} = U_{\mathcal{F}_{\tilde{T}^*}} (\bigoplus_{i \in I} (\mathcal{H}_{\mathtt{A}_i} \otimes \mathcal{H}_{\mathtt{B}_i}))$ for some unitary $U_{\mathcal{F}_{\tilde{T}^*}}$ (the notation will soon make sense) and an index set $I$ such that \begin{align} \label{eq:mainMessageJapaneseTheorem}
    \hat{V} =  \bigoplus_{i \in I} (\idmat_{\mathtt{A}_i} \otimes V_i)\quad \text{and} \quad  \mathcal{F}_{T} = Q^\dagger U_{\mathcal{F}_{\tilde{T}^*}}\bigoplus_{i \in I} (\blt(\mathcal{H}_{\mathtt{A}_i}) \otimes \sigma_i) U_{\mathcal{F}_{\tilde{T}^*}}^\dagger Q, 
\end{align}
where $\hat{V} = (U_{\mathcal{F}_{\tilde{T}^*}} \otimes \idmat_\mathtt{E})\left[(Q \otimes \idmat_\mathtt{E})V Q^\dagger\right] U_{\mathcal{F}_{\tilde{T}^*}}^\dagger$, all $V_i \in \blt(\mathcal{H}_{\mathtt{B}_i}; \mathcal{H}_{\mathtt{B}_i} \otimes \mathcal{H}_\mathtt{E})$ are isometries and all $\sigma_i \in \blt(\mathcal{H}_{\mathtt{B}_i})$ are density matrices. 

We can reproduce this result as follows: By elementary considerations (exploiting the positivity of $T$), the map $\tilde{T} : \blt(\mathcal{H}_{\mathcal{F}_T}) \rightarrow \blt(\mathcal{H}_{\mathcal{F}_T})$, $\tilde{T}(X) = QT(Q^\dagger X Q)Q^\dagger$ is a trace-preserving CP-map (see \cite[Proof of Lemma 6.4]{WolfQuantumChannels} for details), which, by construction, has a full-rank fixed-point. Since Lindblad \cite[Section 3]{lindblad1999general} we know that if $\tilde{T}$ has a full rank fixed-point, then the set of fixed points of the (unital) dual map $\tilde{T}^*$ forms a vN-algebra, $\mathcal{F}_{\tilde{T}^*}$. As $\mathrm{dim}(\mathcal{H}) < \infty$, $\mathcal{F}_{\tilde{T}^*}$ it is atomic and can be decomposed according to Definition \ref{Defn:NormalFormAtomicAlgebra}. Moreover, $\tilde{T}^*(X) = W^\dagger (X \otimes \idmat_\mathtt{E}) W$ with $W = (Q \otimes \idmat_\mathtt{E})V Q^\dagger$ and $\tilde{T}^*(\mathcal{F}_{\tilde{T}^*}) \subseteq \mathcal{F}_{\tilde{T}^*}$. Hence, $W$ decomposes according to Theorem \ref{theorem:cp-map-algebra-to-algebra-stinespring}. This implies $P_i\tilde{T}^*(P_i^\dagger (X_{\mathtt{A}_i} \otimes \idmat_{\mathtt{B}_i})P_i)P_i^\dagger = \left[A_{ii}^\dagger (X_{\mathtt{A}_i} \otimes \idmat_{\mathtt{F}_{ii}})A_{ii}\right] \otimes \idmat_{\mathtt{B}_i} = X_{\mathtt{A}_i} \otimes \idmat_{\mathtt{B}_i}$ for $i \in I$ and $X_{\mathtt{A}_i} \in \blt(\mathcal{H}_{\mathtt{A}_i})$. Since we can choose $A_{ii}$ minimal, this implies $\mathcal{H}_{\mathtt{F}_{ii}} = \mathbb{C}$ and $A_{ii} = \idmat_{\mathtt{A}_i}$. For $i, j \in I$ with $i \neq j$, we have $P_j\tilde{T}^*(P_i^\dagger (X_{\mathtt{A}_i} \otimes \idmat_{\mathtt{B}_i})P_i)P_j^\dagger = \left[A_{ij}^\dagger (X_{\mathtt{A}_i} \otimes \idmat_{\mathtt{F}_{ij}})A_{ij}\right] \otimes \idmat_{\mathtt{B}_i} = 0$ for all $X_{\mathtt{A}_i} \in \blt(\mathcal{H}_{\mathtt{A}_i})$. Thus $A_{ij} = 0$. In conclusion, we have $W = \sum_{i \in I} P_i^\dagger (\idmat_{\mathtt{A}_i} \otimes U_{ii}) P_i$, with $P_i$ given by Definition \ref{Defn:NormalFormAtomicAlgebra}.
From the definition of $P_i$ this is exactly the first part of Eq.~\eqref{eq:mainMessageJapaneseTheorem}, if we identify $V_i$ with $U_{ii}$. 

For the second part, first note that by construction, $\mathcal{F}_{T} = Q^\dagger \mathcal{F}_{\tilde{T}}Q$. By the Brouwer fixed-point theorem, there exists a density matrix $\sigma_i$ such that $\ptr{\mathtt{E}}{V_i \sigma_i V_i^\dagger} = \sigma_i$. With that choice, it is easy to see that operators of the form of the second part in Eq.~\eqref{eq:mainMessageJapaneseTheorem} are fixed-points. 
But since, as a general property of linear maps on finite-dimensional spaces, the dimension of the fixed-point space of $\tilde{T}$ equals the dimension the fixed-point space of $\tilde{T}^*$, the claim follows. Thus we have arrived at the Koashi-Imoto Theorem.

\section{Conclusion} \label{ConclusionSection}

In this work, we have fully characterized the generators of quantum dynamical semigroups with an invariant vN-subalgebra. We have provided a constructive normal form for such restricted GKLS-generators and determined the freedom in their representation. In particular, these results encompass corresponding characterizations for CP-maps with invariant vN-subalgebras.

The assumption of an invariant atomic vN-subalgebra implies that the restriction of the quantum dynamical semigroup to that subalgebra is again a valid quantum dynamical semigroup. This means that we can also interpret Theorem~\ref{THM:AtomicLindbladNormalForm} as providing, given a GKLS generator on a vN-subalgebra, a complete characterization of the possible extensions to a GKLS generator on $\blt(\mathcal{H})$.
In particular, Theorem~\ref{theorem:cp-map-algebra-to-algebra-stinespring} can be regarded as a constructive version of Arveson's extension theorem~\cite{arveson1969subalgebras, paulsen2002completely}, describing the most general CP extension on $\blt (\mathcal{H})$ of a given CP-map defined on a vN-subalgebra, if that subalgebra is atomic.

As demonstrated in Section~\ref{applicationSection}, our characterization of GKLS-generators with an invariant vN-subalgebra provides a unifying perspective on the results of different prior works. 
We expect that this point of view can be useful for further scenarios, such as the study of dynamical semigroups of higher-order quantum maps~\cite{Chiribella_2008, bisio2019theoretical}, generalizing dynamical semigroups of quantum superchannels~\cite{hasenoehrl2021quantum}.

\paragraph{Acknowledgements:} M.H. and M.C.C. thank Michael M. Wolf, Andreas Bluhm, Li Gao, and Zahra Baghali Khanian for insightful and encouraging discussions.  
\newpage

\bibliography{references}{}
\bibliographystyle{plain}

\newpage
\appendix

\section{Auxiliary Lemma}

In this appendix, we provide a full statement and a complete proof of a lemma useful in proving Theorem~\ref{THM:AtomicLindbladNormalForm}.
Here, we use the notation from Section \ref{Defn:NormalFormAtomicAlgebra}.
\begin{theorem}{Lemma} \label{lem:HalfCommutingBsLemma}
Let $\mathcal{A} \subseteq \blt(\mathcal{H})$ be an atomic weakly closed *-algebra. An operator $B \in \blt(\mathcal{H} \otimes \mathcal{H}_{\tilde{\mathtt{E}}}; \mathcal{H} \otimes \mathcal{H}_\mathtt{E})$ satisfies $(X_\mathcal{A} \otimes \idmat_{\mathtt{E}})B = B (X_\mathcal{A} \otimes \idmat_{\tilde{\mathtt{E}}})$ for all $X_\mathcal{A} \in \mathcal{A}$ if and only if there exist $B_0 \in \blt(\mathcal{H}_0 \otimes \mathcal{H}_{\tilde{\mathtt{E}}}; \mathcal{H}_0 \otimes \mathcal{H}_\mathtt{E})$ and for every $i \in I$ an operator $B_i \in \blt(\mathcal{H}_{\mathtt{B}_i} \otimes \mathcal{H}_{\tilde{\mathtt{E}}}; \mathcal{H}_{\mathtt{B}_i} \otimes \mathcal{H}_\mathtt{E})$ such that $B$ is given by the SOT-convergent series 
\begin{equation}
    B = (P_0^\dagger \otimes \idmat_\mathtt{E}) B_0 (P_0 \otimes \idmat_{\tilde{\mathtt{E}}}) + \sum_{i \in I} (P_i^\dagger \otimes \idmat_\mathtt{E}) (\idmat_{\mathtt{A}_i} \otimes B_i) (P_i \otimes \idmat_{\tilde{\mathtt{E}}}).\nonumber
\end{equation}
\end{theorem}
\begin{proof}
We fix orthonormal bases $\{\ket{e_n}\}_{n \in N}$ and $\{\ket{\tilde{e}_m}\}_{m \in M}$ of $\mathcal{H}_\mathtt{E}$ and $\mathcal{H}_{\tilde{\mathtt{E}}}$, respectively. Since $(X_\mathcal{A} \otimes \idmat_\mathtt{E}) B = B(X_\mathcal{A} \otimes \idmat_{\tilde{\mathtt{E}}})$ for all $X_\mathcal{A} \in \mathcal{A}$, it follows that the operators $\beta_{n m} := (\idmat_\mathcal{H} \otimes \bra{e_n})B(\idmat_\mathcal{H} \otimes \ket{\tilde{e}_m})$ belong to $\mathcal{A}^\prime$. Thus (following the discussion around Eq.~\eqref{eq:FormOfAtomicCommutant}), there are operators $\beta_{n m;0} \in \blt(\mathcal{H}_0)$ and $\beta_{n m;i} \in \blt(\mathcal{H}_{\mathtt{B}_i})$ such that $\beta_{n m} = P_0^\dagger \beta_{n m;0} P_0 + \sum_{i \in I} P_i^\dagger (\idmat_{\mathtt{A}_i} \otimes \beta_{n m;i}) P_i$. We define $B_0 = \sum_{n \in N, m \in M} (\idmat_{0} \otimes \ket{e_n})\beta_{n m;0}(\idmat_{0} \otimes \bra{\tilde{e}_m}) \in \blt(\mathcal{H}_0 \otimes \mathcal{H}_{\tilde{\mathtt{E}}}; \mathcal{H}_0 \otimes \mathcal{H}_\mathtt{E})$ and for all $i \in I$ the operator $B_i = \sum_{n \in N, m \in M} (\idmat_{\mathtt{B}_i} \otimes \ket{e_n})\beta_{n m;i}(\idmat_{\mathtt{B}_i} \otimes \bra{\tilde{e}_m}) \in \blt(\mathcal{H}_{\mathtt{B}_i} \otimes \mathcal{H}_{\tilde{\mathtt{E}}}; \mathcal{H}_{\mathtt{B}_i} \otimes \mathcal{H}_\mathtt{E})$. We then have
\begin{align*}
    B &= \sum_{n \in N, m \in M } (\idmat_{\mathcal{H}} \otimes \ket{e_n}) \beta_{n m} (\idmat_{\mathcal{H}} \otimes \bra{\tilde{e}_m})\\
    &= \sum_{n \in N, m \in M} (\idmat_{\mathcal{H}} \otimes \ket{e_n}) P_0^\dagger \beta_{n m;0} P_0 (\idmat_{\mathcal{H}} \otimes \bra{\tilde{e}_m}) \\
    &\;\qquad+  \sum_{\substack{n \in N, m \in M,\\i \in I}} (\idmat_{\mathcal{H}} \otimes \ket{e_n}) P_i^\dagger (\idmat_{\mathtt{A}_i} \otimes \beta_{n;i}) P_i(\idmat_{\mathcal{H}} \otimes \bra{\tilde{e}_m}) \\
    &= \sum_{n \in N, m \in M} (P_0^\dagger \otimes \idmat_\mathtt{E}) (\idmat_{0} \otimes \ket{e_n}) \beta_{n m;0} (\idmat_{0} \otimes \bra{\tilde{e}_m}) (P_0 \otimes \idmat_{\tilde{\mathtt{E}}}) \\
    &\;\qquad+ \sum_{\substack{n \in N, m \in M,\\ i \in I}} (P_i^\dagger \otimes \idmat_\mathtt{E}) \left(\idmat_{\mathtt{A}_i} \otimes \left[(\idmat_{\mathtt{B}_i} \otimes \ket{e_n}) \beta_{nm;i}(\idmat_{\mathtt{B}_i} \otimes \bra{\tilde{e}_m})\right]\right) (P_i \otimes \idmat_{\tilde{\mathtt{E}}}) \\ 
    &= (P_0^\dagger \otimes \idmat_\mathtt{E}) B_0 (P_0 \otimes \idmat_{\tilde{\mathtt{E}}}) + \sum_{i \in I} (P_i^\dagger \otimes \idmat_\mathtt{E}) (\idmat_{\mathtt{A}_i} \otimes B_i) (P_i \otimes \idmat_{\tilde{\mathtt{E}}}).
\end{align*}
This is the claimed result. 
\end{proof}

\section{Proof of Theorem \ref{thm:GeneralUniquenessTheorem}} \label{ProofUniquenessAppendix}

\begin{theorem}{Lemma} \label{Lem:KFreedom}
Two operators $K, \tilde{K} \in \blt(\mathcal{H}_\mathtt{A} \otimes \mathcal{H}_\mathtt{B})$ satisfy a) $(X_\mathtt{A} \otimes \idmat_\mathtt{B})K^\dagger + K(X_\mathtt{A} \otimes \idmat_\mathtt{B}) = (X_\mathtt{A} \otimes \idmat_\mathtt{B})\tilde{K}^\dagger + \tilde{K}(X_\mathtt{A} \otimes \idmat_\mathtt{B})$ for all $X_\mathtt{A} \in \blt(\mathcal{H}_\mathtt{A})$ if and only if there exists a self-adjoint $H_\mathtt{B} \in \blt(\mathcal{H}_\mathtt{B})$ such that b) $\tilde{K} = K + \idmat_\mathtt{A} \otimes iH_\mathtt{B}$. 
\end{theorem}
\begin{proof}
If $\tilde{K} = K + \idmat_\mathtt{A} \otimes iH_\mathtt{B}$, then $a)$ holds trivially.  
For the converse, decompose $K$  and $\tilde{K}$ into real and imaginary part as $K = R + iH$ and $\tilde{K} = \tilde{R} + i\tilde{H}$. By choosing $X_\mathtt{A} = \idmat_\mathtt{A}$, we obtain $\tilde{R} = R$. The relation $a)$ the simplifies to $iH (X_\mathtt{A} \otimes \idmat_\mathtt{B}) - (X_\mathtt{A} \otimes \idmat_\mathtt{B})iH = i\tilde{H} (X_\mathtt{A} \otimes \idmat_\mathtt{B}) - (X_\mathtt{A} \otimes \idmat_\mathtt{B})i\tilde{H}$. In terms of the commutator, this reads $\left[i(H - \tilde{H}), X_\mathtt{A} \otimes \idmat_\mathtt{B} \right] = 0$ for all $X_\mathtt{A} \in \blt(\mathcal{H}_\mathtt{A})$. Thus $i(H - \tilde{H})$ is in the commutant of $\blt(\mathcal{H}_\mathtt{A}) \otimes \idmat_\mathtt{B}$, which is $\idmat_\mathtt{A} \otimes \blt(\mathcal{H}_\mathtt{B})$. Hence, there exists a self-adjoint $H_\mathtt{B} \in \blt(\mathcal{H}_\mathtt{B})$ such that $\tilde{H} = H + \idmat_\mathtt{A} \otimes H_\mathtt{B}$. With $\tilde{K} = R + i\tilde{H}$, the claim follows. 
\end{proof}

We provide a proof of Theorem \ref{thm:GeneralUniquenessTheorem}. The proof follows Chapter 30 in \cite{k1992introduction}, in particular the proof of Proposition 30.14 therein.

\begin{proof}
For any triplets $(\mathcal{H}_\mathtt{E}, V, K)$ and $(\mathcal{H}_{\tilde{\mathtt{E}}}, \tilde{V}, \tilde{K})$, we introduce the shorthand $\pi(X) := (X \otimes \idmat_\mathtt{E})V - VX$ and $\tilde{\pi}(X) := (X \otimes \idmat_{\tilde{\mathtt{E}}})\tilde{V} - \tilde{V}X$. For the triplet $(\mathcal{H}_{\tilde{\mathtt{E}}}, \tilde{V}, \tilde{K})$ given in the statement of the theorem, the space $\tilde{S} := \mathrm{span}\{ \tilde{\pi}(X)\ket{\psi} \, | \, X \in \blt(\mathcal{H}), \ket{\psi} \in \mathcal{H} \}$ is invariant under the action of $Y \otimes \idmat_{\tilde{\mathtt{E}}}$ for all $Y \in \blt(\mathcal{H})$, since for any $\tilde{S} \ni \ket{\phi} = \sum_i \tilde{\pi}(X_i)\ket{\psi_i}$, we have 
\begin{align*}
    (Y \otimes \idmat_{\tilde{\mathtt{E}}})\ket{\phi} &= \sum_i (Y \otimes \idmat_{\tilde{\mathtt{E}}}) \tilde{\pi}(X_i)\ket{\psi_i} \\
    &= \sum_i \underbrace{\tilde{\pi}(YX_i)\ket{\psi_i}}_{\in \tilde{S}} - \underbrace{\tilde{\pi}(Y)X_i\ket{\psi_i}}_{\in \tilde{S}} \in \tilde{S}. 
\end{align*}
Thus, the closure of $\tilde{S}$ is of the form $\mathcal{H} \otimes \mathcal{H}_\mathtt{E}$, for some subspace $\mathcal{H}_\mathtt{E} \subseteq \mathcal{H}_{\tilde{\mathtt{E}}}$. Denote by $P \in \blt(\mathcal{H}_{\tilde{\mathtt{E}}}; \mathcal{H}_{\mathtt{E}})$ the associated orthogonal projection onto $\mathcal{H}_{\mathtt{E}}$. We define $V := (\idmat \otimes P)\tilde{V}$. By construction, $V$ satisfies \ref{eq:GKLSMinimality}. So, to prove the first part of the theorem, it remains to construct a suitable $K$. Since $\idmat \otimes P^\dagger P$ is the projection onto the closure of $\tilde{S}$, we obtain
\begin{align*}
    (X \otimes \idmat_{\tilde{\mathtt{E}}})\tilde{V} - \tilde{V}X &= (\idmat \otimes P^\dagger P) ((X \otimes \idmat_{\tilde{\mathtt{E}}})\tilde{V} - \tilde{V}X) \\&= (\idmat \otimes P^\dagger) ((X \otimes \idmat_{\mathtt{E}})V - VX)
\end{align*}
for all $X \in \blt(\mathcal{H})$. A rearrangement yields
\begin{align*}
    (X \otimes \idmat_{\tilde{\mathtt{E}}})(\tilde{V} - (\idmat \otimes P^\dagger)V) = (\tilde{V} - (\idmat \otimes P^\dagger)V)X
\end{align*}
By Lemma \ref{lem:HalfCommutingBsLemma} (with $\mathcal{A} = \blt(\mathcal{H}$)), this implies that there exists $\ket{\tilde{\phi}} \in \mathcal{H}_{\tilde{\mathtt{E}}}$ such that $\tilde{V} - (\idmat \otimes P^\dagger)V = \idmat \otimes \ket{\tilde{\phi}}$. We define $K := \tilde{K} - (\idmat \otimes \bra{\tilde{\phi}}) \tilde{V} + \frac{1}{2} \Vert\tilde{\phi}\Vert^2$, and notice that $PP^\dagger = \idmat_{\mathtt{E}}$. Thus,
\begin{align*}
    V^\dagger (X \otimes \idmat_\mathtt{E}) V &= \left[(\idmat \otimes P^\dagger)V\right]^\dagger (X \otimes \idmat_{\tilde{\mathtt{E}}}) \left[(\idmat \otimes P^\dagger)V\right] \\&= \left[\tilde{V} - \idmat \otimes \ket{\tilde{\phi}}\right]^\dagger (X \otimes \idmat_{\tilde{\mathtt{E}}}) \left[\tilde{V} - \idmat \otimes \ket{\tilde{\phi}} \right].
\end{align*}
From here, it is easy to see that $V^\dagger (X \otimes \idmat_\mathtt{E}) V -K^\dagger X - XK = \tilde{V}^\dagger (X \otimes \idmat_{\tilde{\mathtt{E}}}) \tilde{V} -\tilde{K}^\dagger X - X\tilde{K} = L(X)$ for all $X \in \blt(\mathcal{H})$. This is our first claim.

For the second claim suppose that for $(\mathcal{H}_\mathtt{E}, V, K)$, \ref{eq:GKLSUniquenessStatement} and \ref{eq:GKLSMinimality} are satisfied, and that for  $(\mathcal{H}_{\tilde{\mathtt{E}}}, \tilde{V}, \tilde{K})$, \ref{eq:GKLSUniquenessStatement} is satisfied. A direct calculation reveals that
\begin{align} \label{eq:TildeNotTildeEquivalence}
    \Psi(X, Y) &:= L(X^\dagger Y) - X^\dagger L(Y) - L(X^\dagger)Y + X^\dagger L(\idmat) Y \nonumber\\
    &= \pi(X)^\dagger \pi(Y) = \tilde{\pi}(X)^\dagger \tilde{\pi}(Y),
\end{align}
for all $X, Y \in \blt(\mathcal{H})$.
On $S := \mathrm{span}\{\pi(X)\ket{\psi}\, | \, X \in \blt(\mathcal{H}), \ket{\psi} \in \mathcal{H} \}$, we define a map $W_0$ by linear extension of the relation $W_0 \pi(X) \ket{\psi} := \tilde{\pi}(X) \ket{\psi}$. This is well-defined, since if $\sum_i \pi(X_i)\ket{\psi_i} = \sum_j \pi(Y_j)\ket{\psi_j}$
, then 
\begin{align*}
    &\Vert \sum_i \tilde{\pi}(X_i)\ket{\psi_i}  - \sum_j \tilde{\pi}(Y_j)\ket{\psi_j} \Vert^2 \\
    &= \sum_{i, i^\prime} \braket{\psi_i}{\tilde{\pi}(X_i)^\dagger\tilde{\pi}(X_{i^\prime})  \psi_{i^\prime}} + 
    \sum_{i, j^\prime} \braket{\psi_i}{\tilde{\pi}(X_i)^\dagger\tilde{\pi}(Y_{j^\prime})  \psi_{j^\prime}} \\
    &\hphantom{=~}+ \sum_{j, i^\prime} \braket{\psi_j}{\tilde{\pi}(Y_j)^\dagger\tilde{\pi}(X_{i^\prime})  \psi_{i^\prime}} + \sum_{j, j^\prime} \braket{\psi_j}{\tilde{\pi}(Y_j)^\dagger\tilde{\pi}(Y_{j^\prime})  \psi_{j^\prime}} \\
    &\overset{\eqref{eq:TildeNotTildeEquivalence}}{=} \sum_{i, i^\prime} \braket{\psi_i}{\pi(X_i)^\dagger\pi(X_{i^\prime})  \psi_{i^\prime}} + 
    \sum_{i, j^\prime} \braket{\psi_i}{\pi(X_i)^\dagger\pi(Y_{j^\prime})  \psi_{j^\prime}} \\
    &\hphantom{=~}+ \sum_{j, i^\prime} \braket{\psi_j}{\pi(Y_j)^\dagger\pi(X_{i^\prime})  \psi_{i^\prime}} + \sum_{j, j^\prime} \braket{\psi_j}{\pi(Y_j)^\dagger\pi(Y_{j^\prime})  \psi_{j^\prime}} \\ 
    &= \Vert \sum_i \pi(X_i)\ket{\psi_i}  - \sum_j \pi(Y_i)\ket{\psi_j} \Vert^2 = 0. 
\end{align*}
Furthermore, $W_0$ can be extended to an isometry $\overline{W}_0$ on the closure of $S$, since for any $\ket{\phi} = \sum_i \pi(X_i) \ket{\psi_i} \in S$, we have 
\begin{align*}
    \norm{W_0\ket{\phi}}^2 &= \sum_{i, i^\prime} \braket{\psi_i}{\tilde{\pi}(X_i)^\dagger \tilde{\pi}(X_{i^\prime}) \psi_{i^\prime}} \\ &\overset{\eqref{eq:TildeNotTildeEquivalence}}{=}
    \sum_{i, i^\prime} \braket{\psi_i}{\pi(X_i)^\dagger \pi(X_{i^\prime}) \psi_{i^\prime}} = \norm{\phi}^2 .
\end{align*}
Moreover, from 
\begin{align*}
    (X \otimes \idmat_{\tilde{\mathtt{E}}})\overline{W}_0 \pi(Y)\ket{\psi} &=  (X \otimes \idmat_{\tilde{\mathtt{E}}})\tilde{\pi}(Y)\ket{\psi} = \tilde{\pi}(XY)\ket{\psi} - \tilde{\pi}(X) Y \ket{\psi} \\&= \overline{W}_0 \pi(XY)\ket{\psi} - \overline{W}_0\pi(X) Y \ket{\psi} = \overline{W}_0 (X \otimes \idmat_{\mathtt{E}}) \pi(Y) \ket{\psi}
\end{align*}  
and totality of $S$, we conclude that $(X \otimes \idmat_{\tilde{\mathtt{E}}})\overline{W}_0 = \overline{W}_0 (X \otimes \idmat_{\mathtt{E}})$. Lemma \ref{lem:HalfCommutingBsLemma} (with $\mathcal{A} = \blt(\mathcal{H})$ and the roles of $\mathcal{H}_\mathtt{E}$ and $\mathcal{H}_{\tilde{\mathtt{E}}}$ interchanged) yields that there is an isometry $W \in \blt(\mathcal{H}_\mathtt{E}; \mathcal{H}_{\tilde{\mathtt{E}}})$ such that $\overline{W}_0 = \idmat \otimes W$. We note that $W_0$ maps $S$ surjectively onto $\tilde{S}$. Thus, (since isometries have closed ranges) if $\tilde{S}$ is dense, $\overline{W}_0$ is surjective and hence $W$ is a unitary. This is the claim of the last sentence in the theorem. 

It remains to verify Eq.~\eqref{eq:FreedomSimpleLindbladGenerator}. To this end, note that by definition $(\idmat \otimes W)((X\otimes \idmat_{\mathtt{E}})V - VX) = (X \otimes \idmat_{\tilde{\mathtt{E}}})\tilde{V} - \tilde{V}X$, which can be expressed as
\begin{align*}
    (X \otimes \idmat_{\tilde{\mathtt{E}}}) ((\idmat \otimes W)V - \tilde{V}) = ((\idmat \otimes W)V - \tilde{V})X.
\end{align*}
Since this holds for all $X \in \blt(\mathcal{H})$, Lemma \ref{lem:HalfCommutingBsLemma} (with $\mathcal{A} = \blt(\mathcal{H})$) tells us that there exists a vector $\ket{\tilde{\psi}} \in \mathcal{H}_{\tilde{\mathtt{E}}}$ such that $(\idmat \otimes W)V - \tilde{V} = - \idmat \otimes \ket{\tilde{\psi}}$. This is the first part of Eq.~\eqref{eq:FreedomSimpleLindbladGenerator}. 
To find the relation between $K$ and $\tilde{K}$ we equate versions two of $L(X)$ in \ref{eq:GKLSUniquenessStatement} (with and without the tilde) and substitute $\tilde{V} = (\idmat \otimes W)V + \idmat \otimes \ket{\tilde{\psi}}$  After expanding the quadratic term and some cancellations and rearrangements, we arrive at
\begin{align*}
      \hat{K}^\dagger X + X\hat{K} = \tilde{K}^\dagger X + X\tilde{K},\quad \text {for all } X\in\blt (\mathcal{H}),
\end{align*}
with $\hat{K} = K + (\idmat \otimes \bra{\tilde{\psi}}W)V + \frac{1}{2} \Vert \tilde{\psi} \Vert^2$.
By Lemma \ref{Lem:KFreedom} (with $\mathcal{H}_\mathtt{B} = \mathbb{C}$), there is $\mu \in \mathbb{R}$ such that $\tilde{K} = \hat{K} + i\mu$. This finishes the proof.
\end{proof}

\end{document}